\pdfoutput=1

\documentclass[sigconf,authorversion,nonacm,balance=false]{acmart}

\usepackage[lambda, operators, sets, primitives, keys, adversary, ff, asymptotics, complexity]{cryptocode}
\usepackage{xspace}
\usepackage{multirow}
\usepackage{makecell}
\usepackage{url}
\usepackage{amsthm}
\usepackage{enumitem}
\usepackage[capitalise, nameinlink, noabbrev]{cleveref}
\pcfixcleveref  

\newtheorem{theorem}{Theorem}

\newtheorem{lemma}{Lemma}

\theoremstyle{definition}
\newtheorem{definition}{Definition}

\let\originalleft\left
\let\originalright\right
\renewcommand{\left}{\mathopen{}\mathclose\bgroup\originalleft}
\renewcommand{\right}{\aftergroup\egroup\originalright}

\allowdisplaybreaks[1]

\let\originalpcalgostyle\pcalgostyle
\renewcommand{\pcalgostyle}[1]{\originalpcalgostyle{#1}\xspace}

\newcommand*{\comm}{\pcalgostyle{Comm}\xspace}
\newcommand*{\crh}{\pcalgostyle{CRH}\xspace}
\newcommand*{\setup}{\pcalgostyle{Setup}\xspace}
\newcommand*{\nizkpk}{\ensuremath{\pcalgostyle{NIZK}\text{-}\pcalgostyle{PK}}\xspace}

\newcommand*{\ldp}{\pcalgostyle{LDP}}
\newcommand*{\apply}{\pcalgostyle{Apply}}
\renewcommand{\kgen}{\pcalgostyle{KeyGen}}
\newcommand*{\prove}{\pcalgostyle{Prove}}
\renewcommand{\verify}{\pcalgostyle{Verify}}
\newcommand*{\merkletree}{\pcalgostyle{MerkleTree}}

\renewcommand*{\adv}{\pcadvstyle{A}}

\newcommand*{\cm}{\pckeystyle{cm}\xspace}
\newcommand*{\ek}{\key[ek]\xspace}
\renewcommand*{\pp}{\pckeystyle{pp}\xspace}
\newcommand*{\out}{\pckeystyle{out}\xspace}
\newcommand*{\trap}{\pckeystyle{trap}\xspace}
\newcommand*{\rt}{\pckeystyle{rt}\xspace}

\newcommand{\relation}{\mathcal{R}\xspace}

\begin{document}

\title{Efficient Verifiable Differential Privacy with Input Authenticity in the Local and Shuffle Model}

\author{Tariq Bontekoe}
\orcid{0000-0002-5331-4033}
\affiliation{%
    \institution{University of Groningen}
    \city{Groningen}
    \country{The Netherlands}}
\email{t.h.bontekoe@rug.nl}

\author{Hassan Jameel Asghar}
\orcid{0000-0001-6168-6497}
\affiliation{%
    \institution{Macquarie University}
    \city{Sydney}
    \country{Australia}}
\email{hassan.asghar@mq.edu.au}

\author{Fatih Turkmen}
\orcid{0000-0002-6262-4869}
\affiliation{%
    \institution{University of Groningen}
    \city{Groningen}
    \country{The Netherlands}}
\email{f.turkmen@rug.nl}

\begin{abstract}
    Local differential privacy (LDP) enables the efficient release of aggregate statistics without having to trust the central server (aggregator), as in the central model of differential privacy, and simultaneously protects a client's sensitive data.
    The shuffle model with LDP provides an additional layer of privacy, by disconnecting the link between clients and the aggregator.
    However, LDP has been shown to be vulnerable to malicious clients who can perform both input and output manipulation attacks, i.e., before and after applying the LDP mechanism, to skew the aggregator's results.
    
    In this work, we show how to prevent malicious clients from compromising LDP schemes.
    Our only realistic assumption is that the initial raw input is authenticated; the rest of the processing pipeline, e.g., formatting the input and applying the LDP mechanism, may be under adversarial control.
    We give several real-world examples where this assumption is justified.
    Our proposed schemes for verifiable LDP (VLDP), \emph{prevent both input and output manipulation attacks} against generic LDP mechanisms, requiring \emph{only one-time interaction between client and server}, unlike existing alternatives~\cite{movsowitzdavidowPrivacyPreservingTransactionsVerifiable2023, katoPreventingManipulationAttack2021}.
    Most importantly, we are the first to provide an efficient scheme for \emph{VLDP in the shuffle model}.
    We describe, and prove security of, two schemes for VLDP in the local model, and one in the shuffle model.
    We show that all schemes are \emph{highly practical}, with client run times of less than 2 seconds, and server run times of 5--7 milliseconds per client.
\end{abstract}

\keywords{differential privacy, shuffle model, verifiable computing}

\maketitle

\section{Introduction}\label{sec:introduction}
Most distributed data sharing applications either assume that the data obtained from a source is honestly obtained via the true input, or deviates arbitrarily from it.
Accordingly, one abstracts these sources as either honest or malicious, with the received data inheriting corresponding labels.
However, we argue that in many practical scenarios, the data processing pipeline at the source, from gathering raw input to formatted data to be sent to a central server, has more structure, which is not captured by such a simple abstraction~\cite{bontekoeVerifiablePrivacyPreservingComputing2024}.
The pipeline consists of several sequential components, that pass information to one another culminating in the final formatted data.
In this case, we may realistically assume that the adversary only controls some components of the source, rather than the entire pipeline.
This makes it possible to verify the validity of the claimed raw input and any subsequent processing on it, if the component receiving the raw input is outside adversarial reach.
Our target scenario is collection and release of data from multiple distributed sources by a central server using differential privacy (DP)~\cite{dworkCalibratingNoiseSensitivity2006}.
More specifically, we focus on the \emph{local}~\cite{kasiviswanathan2011ldp} and \emph{shuffle}~\cite{cheuDistributedDifferentialPrivacy2019} models of DP where the distributed nodes (data holders) do not trust the server (aggregator) with their formatted inputs in the clear.
On the other hand, the server needs to ensure that the data received from the clients is correctly obtained from the true, raw input.
The following \emph{use cases} further motivate the aforementioned threat model.

In \emph{sensor networks}, the main program of a sensor device decides which sensors to read data from and what data to send in a prescribed format to the server.
The sensor device contains various sensors, which are individual hardware components, e.g., chips (see for example~\cite{mbedNAMote72Mbed2018}).
An adversary could (indirectly) take control of the sensor device by replacing this main program with its own malicious program, which may influence the local data processing pipeline.
But, such a program does not corrupt the physical sensor itself, nor the raw data it produces.
    
Consider the case of energy companies obtaining the total (or average) power consumption of a group of households fitted with \emph{smart meters} at regular time intervals.
This data may reveal private information such as the sleeping patterns of house occupants~\cite{memoirs-smart-meter}.
Privacy to individual households can be provided by applying LDP to smart meter readings.
Smart meters do not currently support such functionality, and implementing it in all of them is not cost-efficient.
Fortunately, LDP could be applied via an app outside the smart meter.
However, as this app is outside the trusted environment, it can be manipulated to serve a malicious purpose, i.e., it might output completely different data, and thereby skew statistics.
    
\emph{Smartphones} and \emph{wearables} share their location via GPS sensors, which can be used for crowd estimation to identify hotspots or for crowd control.
Crowd estimation does not require exact GPS coordinates of individual devices; a coarse-grained aggregate location distribution often suffices.
This is an excellent use case of LDP to release a location histogram.
But, naive use of LDP may allow an attacker to send arbitrary locations, skewing the distribution.
Here again we can assume raw GPS coordinates from the sensors as being true (operating system space), but the application sending location information to the server may be malicious (user space).

Many \emph{other applications} fit this narrative, such as smartphone (e.g., accelerometer) or smart home (e.g., temperature) sensors.
Raw values (events) from such sensors are collected in the operating system (OS) space, before the applications running in the user space can process them further for a given task, e.g., gesture detection.
Last but not least, the raw inputs may be generated within a trusted computing module.
We give several concrete examples of trusted components in \cref{subsec:threat-model}.
Thus, we assume a setup where raw data is processed securely and correctly via one component (e.g., a secure enclave, OS space, a hardware module) before another program, possibly adversarial, further processes it before sending the formatted input to a server.

Assuming that the raw data is produced by a \emph{trusted component} at a client enables us to verify it using \emph{digital signatures}.
However, we also need to verify that the same input is used in the rest of the processing pipeline; all without revealing the raw input!

\subsubsection*{Our Contributions} 
\begin{itemize}[leftmargin=*]
    \item We propose three schemes for efficient verifiable LDP (VLDP). Our first (baseline) scheme requires multiple interactive rounds between client and server, similar to comparable works~\cite{katoPreventingManipulationAttack2021,movsowitzdavidowPrivacyPreservingTransactionsVerifiable2023}.
    However, unlike these works, our second scheme reduces the server load by requiring only one such round through randomness expansion techniques.
    Our final scheme further decreases the load on the server while simultaneously offering security in the shuffle model and in the presence of colluding clients.
    
    \item We present the \emph{first} scheme for VLDP in the \emph{shuffle model}~\cite{cheuDistributedDifferentialPrivacy2019}.
    In this model, a trusted shuffler shuffles the locally randomized inputs from the nodes, with the net effect of privacy amplification~\cite{balle2019shuffle}.
    The shuffle model scheme cannot be straightforwardly constructed through our LDP schemes, since the requirements for verifiability and input authenticity on the one hand, and unlinkability (necessary for the shuffle model) on the other hand oppose one another.
    Our VLDP scheme in the shuffle model only adds marginal overhead for the client: approximately 1.8 seconds versus 0.6 and 1.1 seconds for our other schemes.
    
    \item All our schemes protect against both \textit{input} and \textit{output manipulation attacks} as defined in~\cite{katoPreventingManipulationAttack2021}.
    In fact, we are the first to deal with input manipulation attacks, i.e., where an attacker can arbitrarily change the initial input while carrying out the rest of the LDP algorithm faithfully.
    We do this by relying on digital signatures created by a trusted component.
    Similarly, we ensure protection against output manipulation, in which the attacker tries to send arbitrary outputs to the server, by using verifiable randomness and zero-knowledge proofs. 
    
    \item We implement our VLDP schemes for the $k$-ary randomized response ($k$-RR) mechanisms for both histograms and real-valued data as described in~\cite{balle2019shuffle}.
    Unlike existing works, which only target specific randomizers, our schemes can be applied to generic randomizers\@.
    In fact, as long as the randomization used in the LDP mechanism can be approximated using a fixed number of uniformly random bits, our protocol can accommodate it.
    Hence, our solutions can be \emph{extended to many other LDP mechanisms}, e.g., Laplace or Staircase RR~\cite{dworkCalibratingNoiseSensitivity2006,wangLSRRLocalDifferential2022}, as detailed in \cref{subsec:ldp-inside-nizk}.
    \item We implement and evaluate our protocols on two real-world datasets: a smart meter dataset to get the approximate energy consumption per household, and a GPS dataset to obtain the histogram of locations.
    Our results show that the protocols are highly practical and scalable.
    Each client takes a maximum of 2 seconds for a single LDP message, and the server takes less than 7 milliseconds per client.
    Furthermore, the communication cost is only 200--485 bytes per client value (plus a small additional one-time message), as we show in \cref{subsec:conc-apps}.
\end{itemize}

\section{Related Work}\label{sec:related-work}
In alignment with our scope, we restrict our discussion to protocols for verifiable differential privacy~\cite{bontekoeVerifiablePrivacyPreservingComputing2024}, while observing that, to the best of our knowledge, no constructions for verifiable DP in the shuffle model can be found in existing academic literature.

\subsubsection*{Local Model} The earliest work in this area appears to be on cryptographic $k$-RR from Ambainis et al.~\cite{ambainis2004cryptographicRR}.
They do mention an even earlier work by Kikuchi et al.~\cite{kikuchi1999stochastic}, who reinvented the notion of RR for voting and provided cryptographic constructions to protect against cheating voters.
According to~\cite{ambainis2004cryptographicRR}, the schemes from~\cite{kikuchi1999stochastic} are less efficient and provide weaker security guarantees than theirs.

The main security concern addressed by~\cite{ambainis2004cryptographicRR} is the privacy of the server (interviewer).
Namely, the client (respondent) should not know the randomized outcome of her true input, because otherwise the respondent may end the protocol.
To ensure this and to verify that the RR mechanism is correctly computed, they propose several protocols based on oblivious transfer, Pedersen commitments and zero-knowledge proofs.
Randomness in the protocol is guaranteed by ensuring that the commitment parameters evaluate exactly to the probability of the correct or wrong response, requiring this probability to be a rational number.
The communication cost therefore suffers for high precision.
Their threat model is different from ours since we do not require privacy of the outcome of the LDP mechanism, and furthermore, it is not clear how their protocol can be extended to the shuffle model.

Kato et al.~\cite{katoPreventingManipulationAttack2021} extend~\cite{ambainis2004cryptographicRR} to several other variants of LDP ($k$-RR, unary encoding (OUE), local hashing (OLH)). However, their techniques are similar and once again assume that only the output can be manipulated, and the user otherwise uses the true value.
Therefore, they do not ensure that the correct input is being used, which, in our case, can be verified through digital signatures.

The constructions of~\cite{ambainis2004cryptographicRR} and~\cite{katoPreventingManipulationAttack2021} are improved by Song et al.~\cite{songEfficientDefensesOutput2023}, who also use an approach based on random sampling.
A client commits to a vector of entries that corresponds to the distribution of the randomized raw input.
Subsequently, the server requests openings of a subset of these commitments, to obtain a perturbed sample, and obtain sufficient guarantees about the correctness of the received commitment vector.
The authors show how to implement their techniques for $k$-RR and OUE randomizers.
The communication and computation costs are clearly an improvement over prior work, however due to the used approach, the communication costs still suffer for high precision.
Additionally, we do not see a clear way to efficiently extend their work to other randomizers.

In~\cite{munilla2022verDpPolling}, the authors propose LDP with verifiable computing to extract binary attributes from anonymous credentials (e.g., older than 18).
These binary attributes are certified through a third party using signatures, e.g., government or bank issued anonymous credentials.
They do not give details on how this signature verification is done.
They do provide detailed verifiable algorithms for binary RR and the exponential mechanism~\cite{mcsherry2007expmech} to sample attributes in a range (e.g., the age).
Unlike us, they do not provide protocols for $k$-RR, the shuffle model, and their protocols are not scrutinized using rigorous security definitions.

The closest work to ours is from~\cite{movsowitzdavidowPrivacyPreservingTransactionsVerifiable2023} who tackle the problem of releasing an attribute associated with a transaction in a differentially private manner while maintaining the anonymity of a transaction in a blockchain system for cryptocurrencies.
They demonstrate their scheme using binary RR, although they do mention that the scheme can be generalized to non-binary attributes, without giving further details (as we show in \cref{subsec:ldp-inside-nizk}, this is not trivial).
The private attribute is signed by a registration authority, so that it can be verified if the correct input was used in the RR mechanism.
They also check if the random coins used for RR are unbiased, through an interactive protocol between client and server.
Finally, the client provides a \nizk{} proof as a proof of correct application of the RR mechanism.
Our protocol has one element inspired by~\cite{movsowitzdavidowPrivacyPreservingTransactionsVerifiable2023}, namely generating joint randomness.
But, unlike them, our protocols apply to generic LDP randomizers, protect against input manipulation, and provide LDP in the shuffle model.
Moreover, apart from our baseline protocol, we only require a single round of interaction between client and server, which greatly reduces the communication load of both.
This does come at the cost of a more expensive \nizk{} proof for each client, but this trade-off pays off quickly (see \cref{sec:evaluation}).
Additionally, we do not suffer from the latency and scalability drawbacks caused by their dependence on blockchain.

\subsubsection*{Central Model}
The construction from~\cite{narayanVerifiableDifferentialPrivacy2015} is for verifying DP in the central model as opposed to the local model.
The main threat model tackled in their paper is a dishonest analyst who may publish wrong results, banking on the inherent noise in DP, which is different than ours.
The work from~\cite{tsaloliDifferentialPrivacyMeets2023} uses a similar threat model to~\cite{narayanVerifiableDifferentialPrivacy2015}.
Randomness, however, is generated interactively between the curator and a ``reader'' (an entity interested in verifying the claims of the analyst).
The work from~\cite{biswasVerifiableDifferentialPrivacy2023} tackles the problem of verifiable DP in the single curator and multiparty setting.
In the former, one server collects all client inputs.
The server then provides differentially private answers to a third party, the analyst.
In the multiparty setting, multiple servers receive inputs (secret shared) from clients and then compute the differentially private answer to a query that is then presented to the analyst.
Input verification is limited to range checks, rather than verifying their exact values.

Concurrent to our work, Bell et al.~\cite{bellCertifyingPrivateProbabilistic2024} present their construction for verifying private probabilistic mechanisms.
They consider a similar setting to~\cite{biswasVerifiableDifferentialPrivacy2023} but provide stronger security guarantees.
However, their solution only enables the server to answer counting queries and provide differential privacy through additive noise.
Specifically, they show how to use the binomial mechanism, and leave possible extensions to other randomizers for future work.
Contrary to~\cite{bellCertifyingPrivateProbabilistic2024}, we do offer verifiability for generic LDP randomizers.
Our techniques could potentially be extended to their work to support a wider variety of randomizers.

Finally, we note the scheme for confidential proofs of DP training from~\cite{shamsabadiConfidentialDPproofConfidentialProof2024}.
The authors show how to prevent output manipulation in DP-SGD~\cite{abadiDeepLearningDifferential2016} training of machine learning models in the central model.
They use zero-knowledge proofs in combination with commitments to the entire dataset.
However, due to enormous circuit size of a zero-knowledge proof for DP-SGD training, their method relies on interactive schemes with a heavy communication load.
This makes their approach unsuitable for the local model.

\subsubsection*{Other Works}
Another related work to ours, but which does not consider DP, is the ADSNARK system~\cite{backes2015adsnark} for proving computation on authenticated data while maintaining privacy.
Like us, they assume a trusted source that can provide authenticated initial data.
The client is required to compute a function of this data and send the result to the server.
To verify that the client has done the computation correctly, they propose their ADSNARK protocol based on Succinct Non-Interactive Arguments of Knowledge (SNARKs). Unlike us however, their setting is not distributed, and does not consider DP client inputs.
Finally, we would like to point out several works showing the susceptibility of LDP to input manipulation (or data poisoning) attacks~\cite{cheuManipulationAttacksLocal2021, wu2022ldp-poison-key-value, li2023ldp-poison-mean-var, caoDataPoisoningAttacks2021}, which highlight the need for cryptographic solutions for the integrity of initial data and their subsequent processing like our schemes.

\section{Preliminaries and Building Blocks}
\label{sec:preliminaries}
We describe the building blocks used in our protocols with
specific attention to zero-knowledge proofs and differential privacy.

\subsubsection*{PRGs from PRFs}
A pseudo-random function (PRF) family is defined as a family of functions implemented by a key $k \in \mathcal{K}$, where $\mathcal{K}$ is the key space.
A function $\prf(k, x)$ from this family deterministically maps an input $x \in \mathcal{X}$ to an output $y \in \mathcal{Y}$.
For a randomly chosen $k$, $\prf(k, \cdot)$ is indistinguishable from a true random function.
We can use PRFs to construct \emph{pseudo-random number generators (PRGs)}~\cite[Section 4.4]{bonehGraduateCourseApplied2023}: if we wish to sample a random bitstring from $\mathcal{Y}^\ell$, we define $\ell$ distinct, fixed input values $x_1, \ldots, x_\ell \in \mathcal{X}$.
Then, we can define a \prg with seed $k \in \mathcal{K}$ as $\prg(k) = \prf(k, x_1) || \ldots || \prf(k, x_\ell).$
We use this \prg construction to realize our schemes.

\subsubsection*{Commitment Schemes} A commitment scheme $C(x,r)$ takes as input a value $x$ and randomness $r$, and outputs a commitment $\cm$.
The pair $(x,r)$ is called the opening of the commitment.
A secure commitment scheme should be both hiding and binding.
\emph{Binding} means that, given a commitment $C(x,r)$, it should be hard to output a pair $(x',r')$, with $x' \neq x$, such that $C(x',r') = C(x,r)$.
\emph{Hiding} implies that, given two commitments to distinct input values, it should be hard to determine which commitment belongs to which input value, i.e., given $x_0 \neq x_1$ and $\cm_b = C(x_b,r)$, for a random secret bit $b$ and random secret $r$, it should be hard to determine $b$.

\subsubsection*{Digital Signature Schemes} A signature scheme \sig is a 3-tuple of p.p.t.\ algorithms $(\kgen, \sign, \verify)$, where \kgen generates the keys, \sign creates a signature $\sigma$ for a message $m$, and \verify asserts whether $\sigma$ is valid for $m$.
We only require that the signature scheme be secure against \emph{existential forgeries under a chosen message attack ($\textsf{EUF-CMA}$)}, i.e., an adversary \adv{} should not be able to create a valid message-signature pair $(m', \sigma')$ for a \emph{new} message $m' \neq m$.

\subsection{Zero-Knowledge Proofs}\label{subsec:zero-knowledge-proofs}
In our constructions, we rely upon \emph{non-interactive zero-knowledge proofs (NIZKs)}.
NIZKs are used to prove the existence of a secret witness $w$ for a given, public \emph{statement} $x$, such that the pair satisfies some \npol-relation $\relation$, i.e., $(x,w) \in \relation$.
Specifically, we consider NIZKs in the common reference string (CRS) model~\cite{blumNoninteractiveZeroknowledgeIts1988,desantisRobustNoninteractiveZero2001,grothSizePairingBasedNoninteractive2016}, which can be defined as a 4-tuple of p.p.t.\ algorithms $(\setup, \prove, \verify, \simulator)$.
The \setup algorithm generates the evaluation $\ek$ and verification $\vk$ keys (and a simulation trapdoor $\trap$) for a given relation $\relation$.
\prove uses $\ek$ to create a valid proof $\pi$ for a given statement-witness pair $(x,w)$, and \verify uses $\vk$ to assert the correctness of $\pi$ for a given statement $x$.
Finally, \simulator uses the trapdoor $\trap$ and $\ek$ to create a simulated proof for a statement $x$.

A secure NIZK scheme should satisfy the following (informal) properties. \textit{Completeness:} given a true statement, an honest prover should be able to convince an honest verifier. \textit{Soundness:} if the statement is false, no prover should be able to convince the verifier that it is true. \textit{Zero-knowledge:} a proof $\pi$ should reveal no information other than the truth of the public statement $x$, specifically it should leak no information about the witness $w$.
(Note: our constructions only require honest-verifier zero-knowledge.)

However, we are not just interested in knowing that a witness exists, we also want to confirm that the prover \emph{knows} this witness.
Therefore, in the remainder of our work we will look at \emph{NIZK proofs of knowledge (NIZK-PKs)}, which are NIZKs that additionally also satisfy knowledge soundness.
\emph{Knowledge soundness} is a stronger version of soundness that additionally requires the existence of an extractor $\edv_\adv$ that can produce a valid witness given complete access to the adversary $\adv$'s state.
In our implementation, we use zk-SNARKs~\cite{bitanskyExtractableCollisionResistance2012}: \nizkpk{} schemes that are also succinct, i.e., the verifier runs in $\poly[\secpar + |x|]$ time and the proof size is $\poly$.

\subsection{Differential Privacy}\label{subsec:differential-privacy}
Differential privacy (DP)~\cite{dworkCalibratingNoiseSensitivity2006} is a formal way of describing database privacy.
It provides precise measures for how much information about a dataset is leaked by (partial) disclosure through queries on the dataset.
Consider a database $X$ containing $n$ entries from the domain $\mathcal{X}$, i.e., $X \in \mathcal{X}^n$.
We consider two databases $X, X' \in \mathcal{X}^n$ as neighbors, denoted $X \sim X'$, if they differ in exactly one entry.
\begin{definition}[Differential Privacy~\cite{dworkCalibratingNoiseSensitivity2006}]
    A randomized algorithm $\mathcal{M}\colon \mathcal{X}^n \rightarrow \mathcal{Y}$ is $(\epsilon,\delta)$-differentially private, if for all $X \sim X' \in \mathcal{X}^n$ and for all $T \subseteq \mathcal{Y}$, we have $\Pr\left[\mathcal{M}(X) \in T\right] \leq e^\epsilon \Pr\left[M(X') \in T\right] + \delta$.
\end{definition}

Any $(\epsilon,\delta)$-DP randomization algorithm has two highly useful properties, following~\cite{balle2019shuffle}:

\begin{lemma}[Post-processing~\cite{dworkCalibratingNoiseSensitivity2006}]\label{lem:post-processing}
    If $\mathcal{M}$ is $(\epsilon,\delta)$-DP, then for every (deterministic or randomized) $\mathcal{M}'$, $\mathcal{M}' \circ \mathcal{M}$ is also \mbox{$(\epsilon,\delta)$-DP\@.}
\end{lemma}

\begin{lemma}[Sequential composition~\cite{dworkBoostingDifferentialPrivacy2010}]\label{lem:composition}
    If $\mathcal{M}_1,\ldots,\mathcal{M}_n$ are $(\epsilon,\delta)$-DP, then the composed algorithm $\mathcal{M}' = \left(\mathcal{M}_1,\ldots,\mathcal{M}_n\right)$ is $(\epsilon', \delta' + n\delta)$-DP for any $\delta' > 0$ and $\epsilon' = \epsilon(e^\epsilon - 1)n + \epsilon\sqrt{2n\log(1/\delta')}$.
\end{lemma}

\subsubsection*{Shuffle Model}
In the shuffle model, there are $n$ clients, each of whom holds a data entry $x_i \in \mathcal{X}$.
The shuffle model considers three algorithms, following the definitions of~\cite{cheuDistributedDifferentialPrivacy2019}:
\begin{itemize}[leftmargin=*]
\item A randomizer $\mathcal{R}\colon \mathcal{X} \rightarrow \mathcal{Y}$ that takes as input a data entry $x_i$ and outputs a value $\tilde{x}_i \in \mathcal{Y}$.\footnote{We only consider the single-message shuffle model. The more general shuffle model allows for an array of $m$ messages to be output by $\mathcal{M}$.}
\item A shuffler $\mathcal{S}\colon \mathcal{Y}^n \rightarrow \mathcal{Y}^n$ that takes as input a vector of $n$ messages and outputs these in a random order.
Specifically, on input $(\tilde{x}_1, \ldots, \tilde{x}_n)$, $\mathcal{S}$, outputs $(\tilde{x}_{\pi_1},\ldots, \tilde{x}_{\pi_n})$, where $\pi$ is a uniform random permutation of $[n]$.
\item An aggregator, or analyst, $\mathcal{C}\colon \mathcal{Y}^n \rightarrow \mathcal{Z}$, that takes as input a vector of $n$ messages $(\tilde{x}_{\pi_1},\ldots, \tilde{x}_{\pi_n})$ and outputs an estimation of $f(x_1,\ldots,x_n)$.
\end{itemize}

\begin{definition}[DP in the Shuffle Model~\cite{cheuDistributedDifferentialPrivacy2019}] A protocol $(\mathcal{R}, \mathcal{S}, \mathcal{C})$ is $(\epsilon,\delta)$-DP if the protocol $\mathcal{S}(\mathcal{R}(x_1),\ldots,\mathcal{R}(x_n))$ is $(\epsilon,\delta)$-DP\@.
\end{definition}

As a consequence of \cref{lem:composition} and \cref{lem:post-processing}, there is a composition property equivalent to \cref{lem:composition} for the shuffle model~\cite{cheuDistributedDifferentialPrivacy2019}.

\subsubsection*{Local Differential Privacy (LDP)}
When one replaces the shuffler $\mathcal{S}$ by an identity function, i.e., the vector of messages is not shuffled, we are left with the well-known model for LDP~\cite{kasiviswanathan2011ldp}:
\begin{definition}[Local Differential Privacy] A randomized algorithm $\mathcal{R}: \mathcal{X} \rightarrow \mathcal{Y}$ is $(\epsilon,\delta)-LDP$, if for all pairs $x,x' \in \mathcal{X}$, and for all $T \subseteq \mathcal{Y}$, we have $\Pr\left[\mathcal{R}(x) \in T\right] \leq e^\epsilon \Pr\left[\mathcal{R}(x') \in T\right] + \delta$.    
\end{definition}
The purpose of the shuffle mechanism is to amplify the privacy achievable via LDP. We give a concrete example of this in \cref{sec:dp-algorithms}.
In the remainder of this work, when we refer to an LDP algorithm, we will only denote the local randomizer, unless stated otherwise.

\section{DP Algorithms}\label{sec:dp-algorithms}

We consider two LDP algorithms, both of which appear in~\cite{balle2019shuffle}.
The first locally randomizes a real-number input $x \in [0, 1]$.
The goal of the aggregator is to output the sum of these inputs from $n$ users.
The second algorithm takes as input an integer $x \in [k]$ for $k \ge 2$, and locally randomizes it.
The application in this case is a histogram of values in $[k]$.
The algorithms are shown in \cref{fig:ldp-algos}.

\begin{figure}
    \begin{pchstack}[boxed,center]
        \scriptsize
        \pseudocode[
            head={LDP Algorithm for Reals~\cite{balle2019shuffle}},
            codesize=\scriptsize
        ]{
            \text{\textbf{input: }} k \in \mathbb{N}, x \in [0, 1], \gamma \in [0, 1]\\
            \overline{x} \gets \lfloor xk \rfloor + \text{Ber}(xk - \lfloor xk \rfloor)\\
            b \gets \text{Ber}(\gamma) \\
            \pcif b = 0 \pcdo \\
            \t \tilde{x} \gets \overline{x}\\
            \pcelse \\
            \t \tilde{x} \sample \{0, 1, \ldots, k\}\\
            \pcreturn \tilde{x}
        }
        \;\;
        \pseudocode[
            head={LDP Algorithm for Histograms~\cite{balle2019shuffle}},
            codesize=\scriptsize
        ]{
            \text{\textbf{input: }} k \in \mathbb{N}, x \in [k], \gamma \in [0, 1]\\
            b \gets \text{Ber}(\gamma) \\
            \pcif b = 0 \pcdo \\
            \t \tilde{x} \gets x\\
            \pcelse \\
            \t \tilde{x} \sample \{1, \ldots, k\}\\
            \pcreturn \tilde{x}
        }
    \end{pchstack}
    \Description{Details of the LDP algorithms for reals and histograms.}
    \caption{LDP randomizers for reals and histograms.}
    \label{fig:ldp-algos}
\end{figure}

In the LDP algorithm for reals, without loss of generality, we assume $x \in [0, 1]$.
For a precision level $k$, we first encode $x$ as an integer as $\overline{x} = \lfloor xk \rfloor + \text{Ber}(xk - \lfloor xk \rfloor)$~\cite{balle2019shuffle}.
It is easily verified that this encoding ensures that $\mathbb{E}(\overline{x}/k)$, which is the expected value of the decoded $\overline{x}$, is exactly $\mathbb{E}(x)$.
This makes the range of $\overline{x}$ equal to $\{0, 1, \ldots, k\}$.
This algorithm is $\epsilon$-DP, as long as $\frac{1 - k\gamma/(k+1)}{\gamma/(k+1)} \leq e^{\epsilon}$.
Equating the left hand side to the right hand side, we get
$\gamma = \frac{k + 1}{e^{\epsilon} + k}$.

Thus, we can set $\gamma$ to this value given $\epsilon$ and $k$.
If $\mathcal{R}$ is $(\epsilon, \delta)$-LDP, then the mechanism $\mathcal{M}: \mathcal{X}^n \rightarrow \mathcal{Y}^n$ defined as $\mathcal{M}(x_1, \ldots, x_n) = \mathcal{R}^n = (\mathcal{R}(x_1), \ldots, \mathcal{R}(x_n))$ is also $(\epsilon, \delta)$-DP\@.

For the LDP algorithm for histograms we can determine $\gamma$ using the same equation as above. However, we need to replace each occurrence of $k$ by $k-1$, due to the different range for $x$.

\subsubsection*{Aggregator} The aggregator for the LDP histogram algorithm simply outputs the histogram, i.e., the number of inputs for each $i \in [k]$.
For the LDP algorithm for reals, the aggregator should de-bias first.
Let $x_i$ be the $i$-th user's input, let $\overline{x}_i$ be the same input with precision $k$, and $\tilde{x}_i$ the $i$-th user's output of the LDP algorithm.
Then, as shown in \cref{subsec:de-bias}, the aggregator outputs $
\frac{1}{1-\gamma} \left( \frac{\sum_{i = 1}^n \tilde{x}_i}{k} - \frac{\gamma n}{2} \right)$, as estimate of $\frac{1}{k}\sum_{i = 1}^n \overline{x}_i$~\cite{balle2019shuffle}, which itself estimates $\sum_{i = 1}^n x_i$.

\subsubsection*{Shuffle Model} In the shuffle model, the inputs from all parties are first shuffled randomly, and then given to the aggregator.
This results in privacy amplification, as the aggregator now does not know which input belongs to which user.
The shuffle model of DP employs a shuffler $\mathcal{S}\colon \mathcal{Y}^n \rightarrow \mathcal{Y}^n$, which is a random permutation of its inputs.
The algorithm $\mathcal{M} \coloneq \mathcal{S} \circ \mathcal{R}^n\colon \mathcal{X}^n \rightarrow \mathcal{Y}^n$ then provides $(\epsilon, \delta)$-DP against the curator, but with the advantage that the local randomizer $\mathcal{R}$ need only be $(\epsilon_0, 0)$-LDP, with $\epsilon_0$ greater than $\epsilon$.
Ignoring logarithmic terms, $\epsilon_0$ is proportional to $n$ and inversely proportional to $\delta$.
Given a value of $\epsilon, \delta$ and $n$, we can use the script provided by~\cite{balle2019shuffle} to obtain a value of $\epsilon_0$ which uses a tighter analysis than given by the implicit bounds in the paper.
For instance, for the LDP histogram mechanism described above with $k=10$, i.e., $k$-ary RR, with $n = 100$ participants, $\delta = 10^{-6}$ and $\epsilon = 0.1$, we get $\epsilon_0 \approx 1.0032$ through the Bennett bound.
Thus, we can use the mechanism \emph{10 times more} than the LDP mechanism alone.

\subsection{LDP inside NIZK}\label{subsec:ldp-inside-nizk}
To verify the above LDP algorithms inside a \nizk circuit, we need to define how we evaluate an LDP algorithm in a \emph{deterministic} fashion, given a \emph{fixed} number of uniform random bits.
It must be deterministic in the sense that we need to be able to `recreate' random sampling inside the \nizk circuit.
Moreover, we observe that the \nizk proof is computed over a given, fixed, agreed upon relation $\relation$.
Therefore, the (maximum) number of random bits used should be fixed and known up front.
This has the downside that we cannot sample exactly from each distribution, but rather need to sample from \emph{approximate} distributions.
We tackle both issues simultaneously, by defining how to use a uniform random bitstring $\rho$ of length $\ell$, such that the distribution of $\ldp.\apply(x;\rho)$ is statistically close to the true randomized LDP algorithm.

\subsubsection*{Our Approximate Sampling Methods}
Approximations for NIZK encoding of LDP randomizers can be designed for most well-known distributions.
For the LDP randomizers defined in \cref{fig:ldp-algos}, we only need to approximate the Bernoulli distribution and the Discrete Uniform distribution.
\cref{fig:approximate-distributions} shows two algorithms for sampling from these distributions.
We give an additional example in \cref{subsec:example-ldp-nizk}.

These sampling methods clearly match our requirements, and are also statistically close to the true distributions, with the statistical distance decreasing exponentially in $\ell$.
These approximate sampling algorithms replace the random sampling steps in \cref{fig:ldp-algos}.
We give an exact specification of the resulting algorithms in \cref{subsec:approximate-ldp-randomizers}.
For a sufficiently large bitsize of $\rho$, these algorithms are statistically close approximations of the true LDP algorithms.
Moreover, the approximation error decreases exponentially in $|\rho|$.

\begin{figure}
    \begin{pchstack}[boxed,center]
        \scriptsize
        \pseudocode[
        head={$\widetilde{\text{Ber}}(\gamma; \rho)$},
        codesize=\scriptsize
        ]{
            \text{\textbf{input: }} \gamma \in [0, 1], \rho \in \bin^\ell \\
            \text{Interpret $\rho$ as an integer} \\
            \pcif \rho \leq \lfloor \gamma \cdot (2^\ell - 1) \rfloor \pcdo \\
            \t b \gets 1 \\
            \pcelse \\
            \t b \gets 0 \\
            \pcreturn b
        }
        \;\;
        \pseudocode[
        head={$\widetilde{\text{Unif}}([lb, ub]; \rho)$},
        codesize=\scriptsize
        ]{
            \text{\textbf{input: }} lb, ub \in \mathbb{Z}: lb < ub, \rho \in \bin^\ell \\
            \text{Interpret $\rho$ as an integer} \\
            \Delta \gets \lfloor 2^\ell/(ub - lb + 1) \rfloor \\
            \pcfor j \pcin \{0,\ldots,ub - lb - 1\} \\
            \t \pcif j \cdot \Delta \leq \rho < (j + 1) \cdot  \Delta \pcdo \\
            \t \t \pcreturn lb + j \\
            \pcreturn ub
        }
    \end{pchstack}
    \Description{Details of the approximate sampling distributions.}
    \caption{Algorithms for approximately sampling from the Bernoulli and Discrete Uniform distribution.}
    \label{fig:approximate-distributions}
\end{figure}

\subsubsection*{Generalization to Other Randomizers}
The previous construction can be generalized to many other distributions, thereby supporting our claim that our schemes are applicable to a wide class of local randomizers.
The essential difference in the construction across randomizers is showing how to efficiently (approximately) sample from the underlying distributions used by the randomizer.
The rest of the adaptations to the \nizk proof are straightforward.
Thus, below we discuss (at a high level) approaches for approximate sampling from other representative or state-of-the-art LDP randomizers to provide further evidence for the feasibility of encoding them inside \nizk circuits.
Additionally, we discuss estimates for $|\rho|$, since this dominates the computation cost (see also \cref{sec:evaluation}).

\begin{itemize}[leftmargin=*]
    \item \emph{Laplace noise} ~\cite{dworkCalibratingNoiseSensitivity2006} is often used in LDP to perturb continuous input values.
    For approximate Laplace sampling, one can first sample a sign bit (by taking the first bit of $\rho$) and then sample an exponentially distributed `distance' $\ell$.
    The latter can be approximated arbitrarily closely by sampling $\ell$ from a Poisson distribution as described in~\cite{dworkOurDataOurselves2006}: one samples $\ell$ in binary, where the $i$-th bit has a predefined bias $\gamma_i$, i.e., it is sampled from $\widetilde{\text{Ber}}(\gamma_i)$.
    The size of $\rho$ for this approach is $|\rho| \approx |\ell| \cdot \text{precision}(\widetilde{\text{Ber}})$, e.g., for 64-bit precision, we have $|\rho| \approx 512$ bytes.

    \item \emph{RAPPOR}~\cite{erlingssonRAPPORRandomizedAggregatable2014} was developed by Google and used as part of the Chrome browser.
    It adds LDP noise to users' responses to questions.
    First it encodes a raw input using a Bloom filter of size $\ell_B$, i.e., by hashing the input to a bit vector using different hash functions for each vector entry.
    Bloom filters generally use simple hash functions, that are evaluated at little cost inside a \nizk circuit.
    The resulting bit-vector is subsequently transformed using Bernoulli random sampling.
    Thus, we can use $\widetilde{\text{Ber}}$ (\cref{fig:approximate-distributions}) to approximate this efficiently.
    The amount of random bits required is $|\rho| \approx \ell_B \cdot \text{precision}(\widetilde{\text{Ber}})$, e.g., for 64-bit precision and a 20-bit Bloom filter, we have $|\rho| \approx 160$ bytes.
    
    \item \emph{Staircase Randomized Response (SRR)}~\cite{wangLSRRLocalDifferential2022} is a recent LDP algorithm for perturbing location data.
    It requires sampling from a `staircase' shaped distribution, where locations close to the true one are more likely to be sampled than locations further away.
    SRR relies on a specific bit-string encoding of locations with length $\ell$ (no more than 46 bits in practice), where the common prefix length is inversely proportional to the distance between two locations.
    Thus, given a true location $x$, we compute three values from $\rho$; the length of the common prefix, the value of the first different bit-pair, and the remaining bits.
    The first two can be sampled using $\widetilde{\text{Unif}}$ (\cref{fig:approximate-distributions}) and the latter can be taken directly from $\rho$.
    The resulting LDP value $\tilde{x}$ is then computed using simple if-else statements.
    We need $|\rho| \approx \ell + 2 \cdot \text{precision}(\widetilde{\text{Unif}})$ for this approach, e.g., for 64-bits precision, we have $|\rho| \approx 22$ bytes.
\end{itemize}

\begin{figure*}
    \centering
    \includegraphics[width=.92\textwidth]{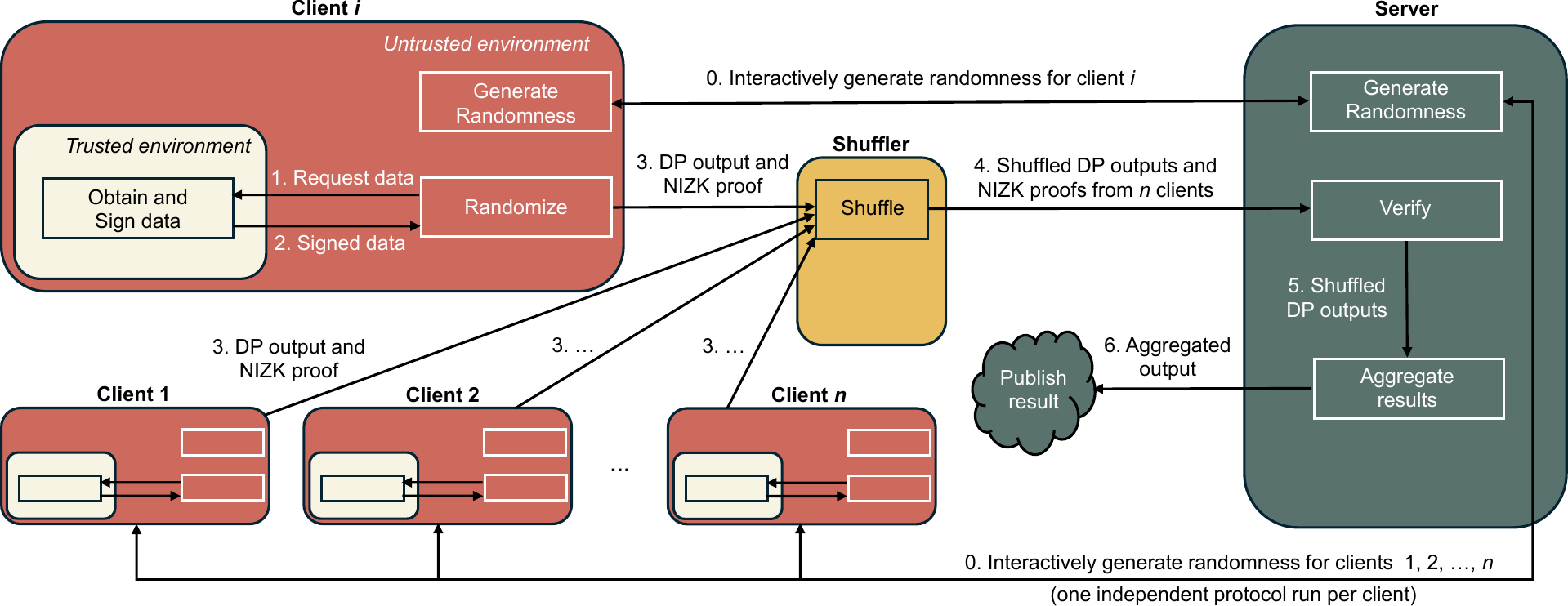}
    \caption{System model for the \pcalgostyle{VLDPPipeline}. For multiple time steps $j$, the clients reiterate the steps as explained further on. When using the `regular' local model, the shuffler is removed and the messages of step 3 are sent directly to the server instead.}
    \label{fig:system_model}
    \Description{System model for the \pcalgostyle{VLDPPipeline} in the shuffle model.}
\end{figure*}

Many more randomizers can be approximated similarly. \emph{Gaussian noise} can be approximated by repeated Bernoulli random sampling~\cite{dworkOurDataOurselves2006}.
Similarly, generic \emph{Randomized Response (RR)}~\cite{warnerRandomizedResponseSurvey1965,kairouzDiscreteDistributionEstimation2016} or \emph{Subset Selection mechanisms}~\cite{yeOptimalSchemesDiscrete2018,wangMutualInformationOptimally2016} are easily computed using (repeated) Bernoulli or Uniform random sampling.
\cite{wangLocallyDifferentiallyPrivate2017} presents a general framework for LDP randomizers for frequency estimation.
This framework splits randomizers in an encoding and a perturbation step.
The listed encodings --- direct, histogram, unary, and local hashing (like in RAPPOR) --- are all easily expressed in \nizk circuit constraints and will have small to negligible overhead.
The listed perturbations are either based on RR or Laplace noise, which we have discussed above. 

Finally, we observe that stateful LDP randomizers (such as~\cite{erlingssonAmplificationShufflingLocal2019}), i.e., randomizers whose behavior depends on previous calls to it, would require the \nizk proof to additionally verify the state update.
Whilst verifying this update is not a problem per se, it would require the previous state as input to the proof, which is not directly supported by our constructions (see \cref{sec:constructions}).
Since the overwhelming majority of LDP randomizers is not stateful, we leave a solution to this challenge as future work.

\section{Verifiable DP in the Local and Shuffle Model}\label{sec:vldp-description}
First, we describe the threat model.
Next, we sketch our system model and give a formal definition for a VLDP scheme that is applicable to both the local and the shuffle model.
Finally, we present formal security definitions.

\subsection{Threat Model}\label{subsec:threat-model}
There are three types of actors in the shuffle model: clients, shuffler and server (see also \cref{fig:system_model}).
The shuffler can be ignored for the local model.
We describe the threat model according to each actor.

\subsubsection{Clients} We assume that all client programs may potentially behave \emph{maliciously}, or collude with other clients, meaning that they could deviate from our scheme arbitrarily, or attempt to use false input data.
However, client programs have no control over the trusted environment and can only obtain signed input data $x$ from it, with signature $\sigma_x = \sig.\sign_{\sk_i}\left(x||t_x\right)$.
Each client $i$'s trusted environment contains a unique secret key $\sk_i$, which cannot be accessed by the potentially malicious client program.
Recall that we use the term trusted environment to denote any controlled environment outside adversarial reach, e.g., secure enclave, OS space, or hardware module.

\paragraph*{Examples of Trusted Environments}
Our model is applicable in situations where a trusted environment is viable.
Given that the presence of trusted environments in consumer hardware is increasingly strong, this is a reasonable requirement.
A concrete example is Apple's Secure Enclave~\cite{appleSecureEnclave}; implemented on iPhones and wearables like Apple watches and HomePods.
The enclave supports EdDSA and ECDSA signatures (see our discussion in \cref{subsec:implementation}). Our protocol minimizes processing within the enclave to input signing only.
Thus, the trusted module can be small.
The rest of the pipeline is executed outside the enclave, and is not assumed to be trustworthy.

Another example is kernel-space vs user-space in Linux-based OSes. 
Apps can only access user-space memory.
Thus any malicious app on a victim phone can not directly access hardware; only indirectly via the kernel~\cite{appleKernelArchitectureOverview}.
We recognize that attacks on trusted environments or kernels (jailbreaks) exist.
Yet, they would also apply to any work based on trusted execution environments (TEEs). Additionally, it is not straightforwardly clear if there is an approach for significantly reducing the reliance on some sort of trusted environment. The idea of loosening or removing this assumption is left to future work.

\subsubsection{Server} We assume the server to be \emph{semi-honest}, i.e., it will not deviate from the scheme, but does try to obtain as much information as possible whilst following the scheme.
Moreover, the server is assumed to be a non-colluding entity.
Finally, we assume that the server can verify which $\pk_i$'s are known/trusted public keys belonging to a trusted environment, e.g., by means of a public key infrastructure or whitelist of trusted keys.

\subsubsection{Shuffler} For the scheme in the shuffle model, we assume that the shuffler is an \emph{honest-but-curious}, independent, \emph{non-colluding} party.
In this work, for the sake of clarity, we will assume that the shuffler is a trusted third party.
In practice, different methods exist for implementing a shuffler, e.g., using mixnets.
We discuss these in \cref{sec:implement-shuffle}.
The actual choice of implementation for the shuffler is out of scope for this work, as our focus lies on constructing efficient, implementation-agnostic, secure VLDP schemes for the local and shuffle model.
Note that it is common for works in the shuffle model to leave this discussion out of scope~\cite{cheuDistributedDifferentialPrivacy2019, balle2019shuffle}.

\subsection{System Model}\label{subsec:system-model}
Let \pcalgostyle{VLDPPipeline} denote the high-level structure of a VLDP scheme, which describes the workings of the VLDP scheme with 1 server and $n$ clients for $T$ time steps (one for each message).
A schematic overview is shown in \cref{fig:system_model}.
First, \pcalgostyle{GenRand} ensures that the client has the necessary inputs to construct verifiably true random values later on.
It can be seen as a sort of preprocessing, where client and server together generate client-specific randomness to be used in the $j$-th time interval $(t_{j-1}, t_j]$, for $j \in [T]$.

Client $i$ generates a VLDP value $\tilde{x}_{i,j}$ for the $j$-th time interval by first requesting a fresh raw input value $x_{i,j}$ from the trusted environment at time $t_x^{i,j} \in (t_{j-1}, t_j]$.
In response, the trusted environment returns a signed input value $x$ with signature $\sigma_x^{i,j} = \sig.\sign_{\sk_i}(x_{i,j}||t_x^{i,j})$, where $\sk_i$ is the secret key of $i$'s trusted environment, which we assume has been generated beforehand.
This signature can be verified using the corresponding $\pk_i$.
Subsequently, $i$ calls \pcalgostyle{Randomize} to verifiably perturb $x_{i,j}$ and obtain $\tilde{x}_{i,j}$ and a correctness proof $\pi_{i,j}$, both of which are sent to the shuffler (or directly to the server in the local model).
The shuffler collects all messages for $(t_{j-1}, t_j]$: $((\tilde{x}_{1,j}, \pi_{1,j}),(\tilde{x}_{2,j}, \pi_{2,j}),\ldots,(\tilde{x}_{n,j}, \pi_{n,j}))$ and forwards these in random order $((\tilde{x}_{?_1, j}, \pi_{?_1, j}),(\tilde{x}_{?_2, j}, \pi_{?_2, j}),\ldots,(\tilde{x}_{?_n, j}, \pi_{?_n, j}))$ to the server, thus ensuring that the server cannot determine which message belongs to which client.

For each received $(\tilde{x}_{?_i, j},\pi_{?_i, j})$, the server runs \pcalgostyle{Verify}, to ensure that $\tilde{x}_{?_i, j}$ is correctly randomized from a value $x$, with $t_x \in (t_{j-1}, t_j]$ signed by a valid $\pk_{?_i}$.
Finally, the server uses all valid values to evaluate and publish its desired output $f(\tilde{x}_{?_1,j}, \ldots, \tilde{x}_{?_n,j})$.

\begin{definition}[VLDP Scheme]
    A VLDP scheme for an LDP algorithm $\ldp.\apply\colon\mathcal{X} \rightarrow \mathcal{Y}$ is a 5-tuple of p.p.t.\ algorithms $\mathcal{VLDP}$ for any number $n \geq 1$ of clients and one prover:
    \begin{itemize}[leftmargin=*]
        \item $\setup(1^\secpar) \rightarrow \pp$: Given the security parameter \secpar, this algorithm returns public parameters \pp.
        This is a tuple containing the \nizk relation $\relation$, parameters of a public key signature scheme $\pp_\text{sig}$ and a commitment scheme $\pp_\text{comm}$.
        Optionally, it also returns a vector $\vec{s}$ of $T$ $\prf$ seeds.
        \item $\kgen(\pp) \rightarrow (\ek, \vk, \pk_s, \sk_s, L)$: Given the public parameters $\pp$, this algorithm returns the evaluation \ek{} and verification key \vk{} for the \nizk proof, and the server's public and secret signature keys $(\pk_s$, $\sk_s)$ along with a list $L$.
        The list $L$ is populated with the identities of clients that have already been processed in a given time interval.
        \item $\pcalgostyle{GenRand}(\pp, \textsf{aux}) \rightarrow \out_c^i$: This interactive protocol between a single client and server takes as input the public parameters \pp{} and optional auxiliary information \textsf{aux}.
        The output of the client is defined as $\out_c^i$, which contains client-generated randomness, commitment to this randomness, server generated randomness, and a server signed signature binding the server generated randomness with the commitment to client's randomness.
        Depending on the scheme's instantiation,  $\out_c^i$ can be used for multiple time intervals or only for one.
        \item $\pcalgostyle{Randomize}(\pp, \ek, t_j, \out_c^i, x, t_x, \sigma_x) \rightarrow (\tilde{x}, \pi, \tau_x)$: Client $i$ uses \pp, $\ek$, a timestamp $t_j$, its output $\out_c^i$ from $\pcalgostyle{GenRand}$, the true input value $x$ with timestamp $t_x$ and signature $\sigma_x$ to compute an LDP value $\tilde{x}$, a \nizk proof $\pi$, and a vector of public values $\tau_x$.
        \footnote{We leave out the index-pair $(i,j)$ for $x, t_x,\sigma_x, \tilde{x}, \pi, \tau_x$ to improve legibility.\label{footnote:left_out_indices}}
        \item $\pcalgostyle{Verify}(\pp, \vk, t_j, \tilde{x}, \pi, \tau_x) \rightarrow \tilde{x} \cup \bot$: The server uses \pp, \vk, a timestamp $t_j$, $\tilde{x}$, $\pi$, and $\tau_x$ to verify whether $\tilde{x}$ was computed honestly.
        If so, it returns $\tilde{x}$, and $\bot$ otherwise.\footref{footnote:left_out_indices}
    \end{itemize}
\end{definition}

\subsection{Security definitions}
\label{subsec:security-defs}
A VLDP scheme should satisfy at least \emph{completeness}, \emph{soundness}, and \emph{zero-knowledgeness}.
Below, we provide the formal definitions of all these properties.
The experiments used in the definitions are detailed in \cref{subsec:experiments}, together with our formal security proofs.

Completeness guarantees that for any authenticated input $x$, created in the right time interval, the output of an honest client will be accepted by an honest server with probability 1.
\begin{definition}[Completeness]
    A scheme $\mathcal{VLDP}$ for an LDP method $\ldp.\apply\colon\mathcal{X} \rightarrow \mathcal{Y}$ with security parameter $\secpar$ is complete if for any $n = \poly$, any $T = \poly$, and for all p.p.t.\ \adv, we have
    $\Pr\left[\bot \in \textbf{\textsf{Exp}}^\text{Comp}_{\adv}(1^\lambda, n, T) \right] \leq \negl$,
    with $\textbf{\textsf{Exp}}^\text{Comp}_{\adv}$ as defined in \cref{fig:completeness-experiment}.
\end{definition}

On the other hand, soundness guarantees that no dishonest client can make a server accept an output, that is not an honest randomization of an authentic input $x$, except with negligible probability.
\begin{definition}[Soundness]\label{def:soundness}
    A scheme $\mathcal{VLDP}$ for an LDP method $\ldp.\apply\colon\mathcal{X} \rightarrow \mathcal{Y}$ with security parameter $\secpar$ is sound if, for any $n = \poly$, any $T = \poly$, for all p.p.t.\ \adv, and $\forall (x_{i,j}, y_{i,j}) \in \mathcal{X} \times \mathcal{Y}$, we have
    \begin{multline*}
    \left\vert\Pr\left[\ldp.\apply(x_{i,j}; \rho_{i,j}) = \{y_{i,j}\}_{i,j} \middle\vert \rho_{i,j} \sample \{0,1\}^* \right]\right. \\
    - \Pr\left[
        \textbf{\textsf{Exp}}^\text{Snd-Real}_{\adv, S^*}(1^\lambda, n, T, \{x_{i,j}\}_{i,j}) = \{y_{i,j}\}_{i,j}
        \middle\vert \right. \\
    \left.\left. \bot \not\in \textbf{\textsf{Exp}}^\text{Snd-Real}_{\adv, S^*}(1^\lambda, n, T, \{x_{i,j}\}_{i,j})
    \right] \right\vert \leq \negl,
    \end{multline*}
    with $\textbf{\textsf{Exp}}^\text{Snd-Real}_\adv$ as defined in \cref{fig:soundness-experiment}, where $S^*$ denotes an honest server that the adversary can interact with.
\end{definition}

The zero-knowledge property guarantees that the server learns nothing about the original input value $x$, other than what could already be learned from its randomization $\tilde{x}$.
\begin{definition}[Zero-knowledge]
    A scheme $\mathcal{VLDP}$ for an LDP method $\ldp.\apply\colon\mathcal{X} \rightarrow \mathcal{Y}$ with security parameter $\secpar$ is zero-knowledge if for any $n = \poly$, any $T = \poly$, there exists a p.p.t.\ simulator $\sdv=(\sdv_1, \sdv_2)$, such that for all p.p.t.\ adversaries \adv, and $\forall (x_{i,j}, y_{i,j}) \in \mathcal{X} \times \mathcal{Y}$, we have
    \begin{multline*}
    \left\{\textbf{\textsf{Exp}}^\text{Zk-Real}_{\adv}(1^\lambda, n, T, \{x_{i,j}\}_{i,j}) = (\cdot ,\{y_{i,j}\}_{i,j})\right\} \\
    \overset{c}{\equiv} \left\{\textbf{\textsf{Exp}}^\text{Zk-Sim}_{\adv,\sdv}(1^\lambda, n, T, \{y_{i,j}\}_{i,j})\right\},
    \end{multline*}
    with $\textbf{\textsf{Exp}}^\text{Zk-Real}_{\adv}$ and $\textbf{\textsf{Exp}}^\text{Zk-Sim}_{\adv,\sdv}$ as defined in \cref{fig:zk-experiment}.
\end{definition}

Additionally, for a VLDP scheme to be secure in the shuffle model, we require \emph{shuffle indistinguishability}, i.e., the server cannot discern an output sent by client $i$ from an output sent by client $i'$.

\begin{definition}[Shuffle indistinguishability]\label{def:shuffle-indistinguishability}
     A scheme $\mathcal{VLDP}$ for an LDP method $\ldp.\apply\colon\mathcal{X} \rightarrow \mathcal{Y}$ with security parameter $\secpar$ has shuffle indistinguishability if for every p.p.t.\ adversary \adv: $2|\Pr[\textbf{\textsf{Exp}}_{\adv}^\text{Ind}(\secpar) = 1 | \textbf{\textsf{Exp}}_{\adv}^\text{Ind}(\secpar) \neq \bot] - \frac{1}{2}| \leq \negl$,
    with $\textbf{\textsf{Exp}}_{\adv}^\text{Ind}$ as defined in \cref{fig:shuffle-ind-experiment}.
\end{definition}

\section{Our Constructions for VLDP}\label{sec:constructions}
In this section, we present our three VLDP schemes and explain the components that together form their respective construction of $\pcalgostyle{VLDPPipeline}$.
Each scheme improves upon the previous one, culminating in an efficient VLDP scheme that can be applied in the shuffle model.
\cref{subsec:proofs} contains formal security analyses.

\begin{enumerate}[leftmargin=*]
    \item The \pcalgostyle{Base} scheme achieves verifiable LDP in the local model.
    Its \pcalgostyle{GenRand} protocol is loosely inspired by the $\pcalgostyle{VerRR}$ algorithm in~\cite{movsowitzdavidowPrivacyPreservingTransactionsVerifiable2023} and should be run once per time interval (and client).
    The other algorithms are novel constructions, which together form a scheme that, unlike~\cite{movsowitzdavidowPrivacyPreservingTransactionsVerifiable2023}, also provides security against input manipulation attacks for authenticated data, supports generic LDP algorithms, and does not require a blockchain.
    \item The \pcalgostyle{Expand} scheme provides the same guarantees, but enables clients to reuse their output $\out_c^i$ of \pcalgostyle{GenRand} for every time interval.
    This significantly decreases the computation and communication load of the server, making the scheme suitable for sequential composition of DP\@.
    \item The \pcalgostyle{Shuffle} scheme has the same communication efficiency as \pcalgostyle{Expand}, but also achieves VLDP in the shuffle model.
\end{enumerate}

\renewcommand{\pclnstyle}[1]{\hspace*{-3pt}\text{\fontsize{6}{7}\selectfont#1}}

\begin{figure}
    
\fontsize{6}{7}\selectfont
\begin{pcvstack}[boxed,center,space=0.1cm]

\procedure[codesize=\fontsize{6}{7}\selectfont]{$\pcalgostyle{VLDPPipeline}_\text{base}$}{
    \pcln \pp \leftarrow \setup_\text{base}(1^\secpar) \\
    \pcln \text{Server computes }(\ek, \vk, \pk_s, \sk_s, L) \leftarrow \kgen_\text{base}(\pp) \\
    \pcln \pcfor \text{Each client $i$ \pcin $\{1,\ldots,n\}$ (in parallel)} \\
    \pcln \t \pcfor j \pcin~\{1,\ldots,T\} \\
    \pcln \t \t \text{Client $i$ obtains } \out_c^{i,j}= \pcalgostyle{GenRand}_\text{base}(\pp, t_j) \\
    \pcln \t \t \text{Client $i$ obtains fresh } (x_{i,j}, t_x^{i,j}, \sigma_x^{i,j}=\sig.\sign_{\sk_i}(x||t_x)) \\
    \pcln \t \t \text{Client $i$ runs } (\tilde{x}_{i,j}, \pi_{i,j}, \tau_{i,j}) = \pcalgostyle{Randomize}_\text{b.}(\pp, \ek, t_j, \out_c^{i,j}, x_{i,j}, t_x^{i,j}, \sigma_x^{i,j}) \\
    \pcln \t \t \text{Server obtains } \tilde{x}_{i,j} = \pcalgostyle{Verify}_\text{base}(\pp, \vk, t_j, \tilde{x}_{i,j}, \pi_{i,j}, \tau_{i,j}) \\
    \pcln \text{Server computes result from all $\tilde{x}_{i,j}$}
}

\begin{pchstack}[space=0.1cm]
\begin{pcvstack}[space=0.1cm]
\procedure[linenumbering,codesize=\fontsize{6}{7}\selectfont]{$\kgen_\text{base}(\pp)$}{
    (\ek, \vk) \leftarrow \nizkpk.\kgen(\relation_\text{base}) \\
    (\sk_s, \pk_s) \leftarrow \sig.\kgen(\pp_\text{sig}) \\
    L \leftarrow \emptyset \\
    \pcreturn (\ek, \vk, \pk_s, \sk_s, L)
}

\procedure[linenumbering,codesize=\fontsize{6}{7}\selectfont]{$\setup_\text{base}(1^\secpar)$}{
    \pp_\text{sig} \leftarrow \sig.\setup(1^\secpar) \\
    \pp_\text{comm} \leftarrow \comm.\setup(1^\secpar) \\
    \vec{t} = (t_0, \ldots, t_T) \\
    \pp = (\relation_\text{base}, \pp_\text{sig}, \pp_\text{comm}, \vec{t}) \\
    \pcreturn \pp
}
\end{pcvstack}

\procedure[codesize=\fontsize{6}{7}\selectfont]{$\relation_\text{base}$}{
    \text{Given $(t_{j-1}, t_j, \pk_i, \cm_{\rho_c}^{i,j}, \rho_s^{i,j}, \tilde{x}^{i,j})$, the} \\
    \text{prover knows } (t_x^{i,j}, x_{i,j}, \sigma_x^{i,j}, \rho_c^{i,j}, r_{\rho_c}^{i,j}) \text{ s.t.:} \\
    \pcln t_x^{i,j} \in (t_{j-1}, t_j] \\
    \pcln \sig.\verify_{\pk_i}(\sigma_x^{i,j}, x_{i,j} || t_x^{i,j}) = 1 \\
    \pcln \cm_{\rho_c}^{i,j} = \comm(\rho_c^{i,j}; r_{\rho_c}^{i,j} ) \\
    \pcln \rho_{i,j} = \rho_c^{i,j} \oplus \rho_s^{i,j} \\
    \pcln \tilde{x}_{i,j} = \ldp.\apply(x_{i,j}; \rho_{i,j})
}
\end{pchstack}

\begin{pchstack}[space=0.1cm]
\procedure[codesize=\fontsize{6}{7}\selectfont]{$\pcalgostyle{GenRand}_\text{base}(\pp, t_j)$ --- \textbf{Client $i$}}{
    \pcln k_c^{i,j} \sample \bin^* \\
    \pcln \rho_c^{i,j} = \prf(k_c^{i,j}, 0) \\
    \pcln r_{\rho_c}^{i,j} \sample \bin^* \\
    \pcln \cm_{\rho_c}^{i,j} = \comm(\rho_c^{i,j}; r_{\rho_c}^{i,j}) \\
    \pcln \text{Send $(\pk_i, \cm_{\rho_c}^{i,j}, t_j)$ to server} \\
    \pcln \text{Receive $(k_s^{i,j}, \sigma_s^{i,j})$ from server} \\
    \pcln \text{If } \sig.\verify_{\pk_s}(\sigma_s^{i,j},\pk_i || \cm_{k_c}^{i,j} || k_s^{i,j} || t_j) \\
    \t \t \t \t \t \t \t \t \t \t \t \t \t \neq 1 \text{, abort} \\
    \pcln \pcreturn \out_c^{i,j} = (\rho_c^{i,j}, r_{\rho_c}^{i,j}, \cm_{\rho_c}^{i,j}, k_s^{i,j}, \sigma_s^{i,j})
}

\procedure[codesize=\fontsize{6}{7}\selectfont]{$\pcalgostyle{GenRand}_\text{base}(\pp, t_j)$ --- \textbf{Server}}{
    \pcln \text{Receive } (\pk_i, \cm_{\rho_c}^{i,j}, t_j) \text{ from client $i$} \\
    \pcln \text{If $\pk_i$ does not belong to $i$, abort} \\
    \pcln \text{If $(i, t_j) \in L$, abort} \\
    \pcln L \leftarrow L \cup \{i, t_j\} \\
    \pcln k_s \sample \bin^* \\
    \pcln \text{Send $(k_s^{i,j}, \sigma_s^{i,j})$ to client $i$} \\
}
\end{pchstack}

\begin{pchstack}[space=0.05cm]
\procedure[codesize=\fontsize{6}{7}\selectfont]{$\pcalgostyle{Randomize}_\text{b.}(\pp, \ek, t_j, \out_c^{i,j}, x_{i,j}, t_x^{i,j}, \sigma_x^{i,j})$}{
    \pcln \rho_s^{i,j} = \prf(k_s^{i,j}, 0) \\
    \pcln \rho_{i,j} = \rho_c^{i,j} \oplus \rho_s^{i,j} \\
    \pcln \tilde{x}_{i,j} = \ldp.\apply(x_{i,j}; \rho_{i,j}) \\
    \pcln \vec{\phi}_{i,j} = (t_{j-1}, t_j, \pk_i, \cm_{\rho_c}^{i,j}, \rho_s^{i,j}, \tilde{x}_{i,j}) \\
    \pcln \vec{w}_{i,j} = (t_x^{i,j}, x_{i,j}, \sigma_x^{i,j}, \rho_c^{i,j}, r_{\rho_c}^{i,j}) \\
    \pcln \pi_{i,j} = \nizkpk.\prove_\ek(\vec{\phi}_{i,j}; \vec{w}_{i,j}) \\
    \pcln \tau_{i,j} = (\pk_i, \cm_{\rho_c}^{i,j}, k_s^{i,j}, \sigma_s^{i,j}) \\
    \pcln \text{Send } (\tilde{x}_{i,j}, \tau_{i,j}) \text{ to server}
}
\procedure[codesize=\fontsize{6}{7}\selectfont]{$\pcalgostyle{Verify}_\text{base}(\pp, \vk, t_j, \tilde{x}_{i,j}, \pi_{i,j}, \tau_{i,j})$}{
    \pcln \text{Parse } \tau_{i,j} = (\pk_i, \cm_{\rho_c}^{i,j}, k_s^{i,j}, \sigma_s^{i,j}) \\
    \pcln \text{If $\pk_i$ does not belong to $i$, abort} \\
    \pcln \text{If } \sig.\verify_{\pk_s}(\sigma_s^{i,j}, \\
    \t \t \t \t \pk_i || \cm_{k_c}^{i,j} || k_s^{i,j} || t_j) \neq 1 \text{, abort} \\
    \pcln \rho_s^{i,j} = \prf(k_s^{i,j}, 0) \\
    \pcln \vec{\phi}_{i,j} = (t_{j-1}, t_j, \pk_i, \cm_{\rho_c}^{i,j}, \rho_s^{i,j}, \tilde{x}_{i,j}) \\
    \pcln \text{If } \nizkpk.\pcalgostyle{Vfy}_\vk(\pi_{i,j}, \vec{\phi}_{i,j}) \neq 1\text{, abort} \\
    \pcln \pcreturn \tilde{x}_{i,j}
}
\end{pchstack}

\end{pcvstack}
\caption{\pcalgostyle{Base} scheme: VLDP with one server and $n$ clients.}
\label{fig:scheme-base}
\Description{Details of the \pcalgostyle{Base} scheme.}
\end{figure}

\renewcommand{\pclnstyle}[1]{\text{\scriptsize#1}}

\subsection{Base Scheme}\label{subsec:base-scheme}
\cref{fig:scheme-base} describes the \textsf{Base} scheme in detail.
Each client $i$ obtains fresh randomness for time interval $j$ by running an independent instance of $\pcalgostyle{GenRand_\text{base}}$ with the server. 
Together, they compute the necessary values to construct a true random value $\rho_{i,j}$ for later use in $\pcalgostyle{Randomize_\text{base}}$.
The bit length of $\rho$ will be equal to the output of the $\prf$ used to generate $\rho$, and is denoted by $|\rho|$.
In case the required number of bits $\ell$ needed to evaluate $\ldp.\apply()$ is lower than $|\rho|$, we can simply ignore the unused bits.
However, in case $\ell > |\rho|$, we need to evaluate the \prf on one or more additional inputs, depending on $\ell$, and concatenate the results.
For clarity, we assume that $\ell \leq |\rho|$ in our scheme definitions, since it can be extended easily using this method.
In our experimental evaluations (\cref{sec:evaluation}), we evaluate the influence of $\ell$ on the performance.

In an instance of $\pcalgostyle{GenRand_\text{base}}$, client $i$ first generates its own random bits $\rho_c^{i,j}$ (we explicitly show the use of a \prf in step 1 and 2 of \cref{fig:scheme-base} to resemble the later schemes).
Subsequently, $i$ computes a commitment $\cm_{\rho_c}^{i,j}$ to $\rho_c^{i,j}$, and shares it along with its trusted environment's public key $\pk_i$, and a time interval marker $t_j$ with the server.
The eventual randomness is then also bound to $t_j$, such that the client cannot create a large batch of random values, and then pick a specific value from this batch.
That would clearly violate the requirements for verifiable randomization.

The server first checks whether $\pk_i$ indeed belongs to $i$, and that $i$ did not previously construct a random value for $t_j$, i.e., whether $(i,t_j) \not\in L$.
Next, the server generates a valid \prf seed $k_s^{i,j}$ and computes $\sigma_s^{i,j} = \sig.\sign_{\sk_s}(\pk_i || \cm_{\rho_c}^{i,j} || k_s^{i,j} || t_j)$.
The server then sends $(k_s^{i,j}, \sigma_s^{i,j})$ to $i$, who verifies $\sigma_s^{i,j}$.
Note that, rather than using a signature, the server could instead maintain a list of $(\pk_i, \cm_{\rho_c}^{i,j})$ for each client and compare this state in $\pcalgostyle{Verify}_\text{base}$ later.
We, however, choose this approach to minimize the server's storage load.

In $\pcalgostyle{Randomize}_\text{base}$, the client computes the server part of the randomness $\rho_s^{i,j}$ from $k_s^{i,j}$, and combines the client and server parts to obtain a true random value $\rho_{i,j} = \rho_c^{i,j} \oplus \rho_s^{i,j}$.
Subsequently, the client uses $\rho_{i,j}$ to transform $x_{i,j}$ into a differentially private value $\tilde{x}_{i,j}$ using $\ldp.\apply()$.
Finally, the client computes the $\nizkpk$ for $\relation_\text{base}$ to attest to a number of statements: (1) the true value $x_{i,j}$ was signed using $\pk_i$ and obtained at a time $t_x^{i,j}$, such that $t_x^{i,j} \in (t_{j-1}, t_j]$; (2) $\rho_{i,j}$ is a true random value, i.e., $\cm_{\rho_c}^{i,j} = \comm(\rho_c^{i,j}; r_{\rho_c}^{i,j})$ and $\rho_{i,j} = \rho_c^{i,j} \oplus \rho_s^{i,j}$; and (3) $\tilde{x}_{i,j}$ is the result of $\ldp.\apply(x_{i,j};\rho_{i,j})$.

When the server receives an LDP value $\tilde{x}^{i,j}$, proof $\pi^{i,j}$ and public values $(\pk_i, \cm_{\rho_c}^{i,j}, k_s^{i,j}, \sigma_s^{i,j})$ from the $i$-th client, it verifies correctness of $\rho_s^{i,j}$, $\sigma_s^{i,j}$ and $\pi_{i,j}$ for $t_j$.
If all hold, it knows that $\tilde{x}^{i,j}$ is a correct DP version of an authentic input.

\renewcommand{\pclnstyle}[1]{\hspace*{-3pt}\text{\fontsize{6}{7}\selectfont#1}}

\begin{figure}
\fontsize{6}{7}\selectfont
\begin{pcvstack}[boxed,center,space=0.1cm]
    
\procedure[codesize=\fontsize{6}{7}\selectfont]{$\pcalgostyle{VLDPPipeline}_\text{expand}$}{
    \pcln \pp \leftarrow \setup_\text{expand}(1^\lambda) \\
    \pcln \text{Server computes }(\ek, \vk, \pk_s, \sk_s, L) \leftarrow \kgen_\text{expand}(\pp) \\
    \pcln \pcfor \text{Each client $i$ (in parallel)} \\
    \pcln \t \text{Client obtains } \out_c^i= \pcalgostyle{GenRand}_\text{expand}(\pp) \\
    \pcln \t \pcfor j \pcin~\{1,\ldots,T\} \\
    \pcln \t \t \text{Client $i$ obtains fresh } (x_{i,j}, t_x^{i,j}, \sigma_x^{i,j}=\sig.\sign_{\sk_i}(x_{i,j}||t_x^{i,j})) \\
    \pcln \t \t \text{Client $i$ runs } (t_j, \tilde{x}_{i,j}, \pi_{i,j}, \tau_{i,j}) = \pcalgostyle{Randomize}_\text{e.}(\pp, \ek, t_j, \out_c^i, x_{i,j}, t_x^{i,j}, \sigma_x^{i,j}) \\
    \pcln \t \t \text{Server obtains } \tilde{x}_{i,j} = \pcalgostyle{Verify}_\text{expand}(\pp, \vk, t_j, \tilde{x}_{i,j}, \pi_{i,j}, \tau_{i,j}) \\
    \pcln \text{Server computes result from all $\tilde{x}_{i,j}$}
}

\begin{pchstack}[space=0.1cm]
\begin{pcvstack}[space=0.1cm]
\procedure[linenumbering,codesize=\fontsize{6}{7}\selectfont]{$\kgen_\text{expand}(\pp)$}{
    (\ek, \vk) \leftarrow \nizkpk.\kgen(\relation_\text{expand}) \\
    (\sk_s, \pk_s) \leftarrow \sig.\kgen(\pp_\text{sig}) \\
    \pp = (\pk_s, \pp_\text{nizk}, \pp_\text{sig}, \pp_\text{comm}, \vec{s}) \\
    L \leftarrow \emptyset \\
    \pcreturn (\ek, \vk, \pk_s, \sk_s, L)
}

\procedure[linenumbering,codesize=\fontsize{6}{7}\selectfont]{$\setup_\text{expand}(1^\secpar)$}{
    \pp_\text{sig} \leftarrow \sig.\setup(1^\secpar) \\
    \pp_\text{comm} \leftarrow \comm.\setup(1^\secpar) \\
    \vec{s} = (s_1, s_2, \ldots, s_T) \sample \bin^{\lambda \times T}\\
    \vec{t} = (t_0, \ldots, t_T) \\
    \pp = (\relation_\text{expand}, \pp_\text{sig}, \pp_\text{comm}, \vec{s}, \vec{t}) \\
    \pcreturn \pp
}
\end{pcvstack}

\procedure[codesize=\fontsize{6}{7}\selectfont]{$\relation_\text{expand}$}{
    \text{Given $(t_{j-1}, t_j, \pk_i, \rt_i, \rho_s^{i,j}, \tilde{x}_{i,j})$,} \\
    \text{the prover knows $(t_x^{i,j}, x_{i,j}, \sigma_x^{i,j},$} \\
    \rho_c^{i,j}, r_{\rho_c}^{i,j}, \cm_{\rho_c}^{i,j}) \text{ such that:} \\
    \pcln t_x^{i,j} \in (t_{j-1}, t_j] \\
    \pcln \sig.\verify_{\pk_i}(\sigma_x^{i,j}, x_{i,j}||t_x^{i,j}) = 1 \\
    \pcln \cm_{\rho_c}^{i,j} = \comm(\rho_c^{i,j}; r_{\rho_c}^{i,j} ) \\
    \pcln \cm_{\rho_c}^{i,j} \text{ is leaf $j$ in \merkletree with root $\rt_i$} \\
    \pcln \rho_{i,j} = \rho_c^{i,j} \oplus \rho_s^{i,j} \\
    \pcln \tilde{x}_{i,j} = \ldp.\apply(x_{i,j}; \rho_{i,j})
}
\end{pchstack}

\begin{pchstack}[space=0.1cm]
\procedure[codesize=\fontsize{6}{7}\selectfont]{$\pcalgostyle{GenRand}_\text{expand}(\pp)$ --- \textbf{Client $i$}}{
    \pcln k_c^i \sample \bin^* \\
    \pcln \vec{\rho}_c^i = ( \prf(k_c^i, 1), \ldots, \prf(k_c^i, T) ) \\
    \pcln \vec{r}_{\rho_c}^i \sample \bin^{T \times *} \\
    \pcln \vec{\cm}_{\rho_c}^i = (\comm(\rho_c^1; r_{\rho_c}^1), \ldots,\comm(\rho_c^T; r_{\rho_c}^T)) \\
    \pcln \rt_i = \merkletree(\vec{\cm}_{\rho_c}^i) \\
    \pcln \text{Send $(\pk_i, \rt_i)$ to server} \\
    \pcln \text{Receive $(k_s^i, \sigma_s^i)$ from server} \\
    \pcln \text{If } \sig.\verify_{\pk_s}(\sigma_s^i, \pk_i || \rt_i || k_s^i) \neq 1 \text{, abort} \\
    \pcln \pcreturn \out_c^i =  (\vec{\rho}_c^i, \vec{r}_{\rho_c}^i, \vec{\cm}_{\rho_c}^i, \rt_i, k_s^i, \sigma_s^i)
}

\procedure[linenumbering, codesize=\fontsize{6}{7}\selectfont]{$\pcalgostyle{GenRand}_\text{expand}(\pp)$ --- \textbf{Server}}{
    \text{Receive $(\pk_i, \rt_i)$ from client $i$} \\
    \text{If $\pk_i$ does not belong to $i$, abort} \\
    \text{If $i \in L$, abort} \\
    L \leftarrow L \cup \{i\} \\
    k_s^i \sample \bin^* \\
    \sigma_s^i = \sig.\sign_{\sk_s}(\pk_i || \rt_i || k_s^i) \\
    \text{Send $(k_s^i, \sigma_s^i)$ to client $i$}
}
\end{pchstack}

\begin{pchstack}[space=0.1cm]
\procedure[linenumbering,codesize=\fontsize{6}{7}\selectfont]{$\pcalgostyle{Randomize}_\text{e.}(\pp, \ek, t_j, \out_c^i, x_{i,j}, t_x^{i,j}, \sigma_x^{i,j})$}{
    \rho_s^{i,j} = \prf(k_s^i||s_j) \\
    \rho_{i,j} = \rho_c^{i,j} \oplus \rho_s^{i,j} \\
    \tilde{x}_{i,j} = \ldp.\apply(x_{i,j}, \rho_{i,j}) \\
    \vec{\phi}_{i,j} = (t_{j-1}, t_j, \pk_i, \rt_i, \rho_s^{i,j}, \tilde{x}_{i,j}) \\
    \vec{w}_{i,j} = (t_x^{i,j}, x_{i,j}, \sigma_x^{i,j}, \rho_c^{i,j}, r_{\rho_c}^{i,j}, \cm_{\rho_c}^{i,j}) \\
    \pi_{i,j} = \nizkpk.\prove_{\ek}(\vec{\phi}_{i,j}; \vec{w}_{i,j}) \\
    \text{Send $( \tilde{x}_{i,j}, \pi_{i,j}, (\pk_i, \rt, k_s^i, \sigma_s))$ to server}
}
\procedure[codesize=\fontsize{6}{7}\selectfont]{$\pcalgostyle{Verify}_\text{expand}(\pp, \vk, t_j, \tilde{x}_{i,j}, \pi_{i,j}, \tau_i)$}{
    \pcln \text{Parse $\tau_i = (\pk_i, \rt_i, k_s^i, \sigma_s^i)$} \\
    \pcln \text{If $\pk_i$ does not belong to $i$, abort} \\
    \pcln \text{If } \sig.\verify_{\sk_s}(\sigma_s^i, \pk_i || \rt_i || k_s^i)\neq 1,\\
    \t \text{abort}\\
    \pcln \rho_s^{i,j} = \prf(k_s^i||s_j) \\
    \pcln \vec{\phi}_{i,j} = (t_{j-1}, t_j, \pk_i, \rt_i, \rho_s^{i,j}, \tilde{x}_{i,j}) \\
    \pcln \text{If } \nizkpk.\pcalgostyle{Vfy}_\vk(\pi_{i,j}, \vec{\phi}_{i,j}) \neq 1 \text{, abort} \\
    \pcln \pcreturn \tilde{x}_{i,j}
}
\end{pchstack}

\end{pcvstack}
\caption{\pcalgostyle{Expand} scheme: only one call to \pcalgostyle{GenRand} per client.}
\label{fig:scheme-expand}
\Description{Details of the \pcalgostyle{Expand} scheme.}
\end{figure}

\renewcommand{\pclnstyle}[1]{\text{\scriptsize#1}}

\subsection{Randomness Expansion (Expand) Scheme}\label{subsec:randomness-expansion-expand-scheme}
The \pcalgostyle{Base} scheme requires one execution of \pcalgostyle{GenRand} for each call to \pcalgostyle{Randomize}, i.e., one per client, per time interval.
Due to the interactive nature of \pcalgostyle{GenRand}, this becomes impractical when the number of clients increases.
The \pcalgostyle{Expand} scheme (\cref{fig:scheme-expand}) uses Merkle trees as compact commitments to multiple random values, to reduce the number of \pcalgostyle{GenRand} executions to only one per client.
Specifically, we update steps 2--4 of \pcalgostyle{GenRand} by creating $T$ commitments to $T$ randomly generated values $\rho_c^{i,j}$, for $j \in [T]$.
Subsequently, we encode all these commitments inside a Merkle tree with root $\rt$ to keep the message size constant and equal to that of $\pcalgostyle{GenRand}_\text{base}$.
The main advantage is that we can now generate $T$ random values with only one round of communication, with communication and server-side cost independent of $T$.

This improvement requires some changes and additional computations for the client in $\pcalgostyle{Randomize}_\text{expand}$.
Following \cref{sec:preliminaries}, given an array of distinct, public values $\vec{s} = (s_1, \ldots, s_T)$, we can define a secure \prg as $\prg(k)\coloneqq \prf(k||s_1)||\ldots||\prf(k||s_T)$.
Thus, if we consider the $j$-th call to $\pcalgostyle{Randomize}_\text{expand}$, we can compute the server part of the randomness (line 1) as $\rho_s^{i,j} = \prf(k_s^i || s_j)$, where $k_s^i$ is the server seed for client $i$.
Observe that the vector $s$ is identical for all clients.
The remainder of $\pcalgostyle{Randomize}$ follows the same structure as in $\textsf{Base}$.
However, we do have to add an additional statement to our \nizkpk{} for $\relation_\text{expand}$, verifying that the client randomness used in the $j$-th call of $\pcalgostyle{Randomize}_\text{expand}$ is indeed the $j$-th entry of the Merkle tree with root $\rt_i$.
This ensures not only that the $i$-th client uses a random value that was committed to before seeing $k_s^i$, but also ensures that $i$ has no choice in which random value in the Merkle tree it uses.
Allowing the client to choose which value it uses could make it possible to influence the value of $\tilde{x}^{i,j}$ for at least one $j$, by cleverly constructing $\vec{\rho}_c$.

\subsection{Shuffle Model Scheme}\label{subsec:shuffle-model}
In both the \pcalgostyle{Base} and \pcalgostyle{Expand} scheme, we consider the regular LDP model. In this model, at time step $j$ the server receives $n$ differentially private values $\tilde{x}_{i,j}$, each of which is directly linkable to the client who sent it.
However, in the shuffle model at time step $j$, the server instead receives a vector $\bf{\tilde{x}}_j$ of $n$ differentially private values, which are not linkable to a particular client.
The server at most knows the group of clients that is collectively responsible for sending this vector of differentially private data.
For simplicity, we assume that there is a trusted shuffler who first collects all the clients' messages and then sends them in random order to the server.
In practice, this may be implemented using, e.g., mixnets (see \cref{sec:implement-shuffle}).

In other words, rather than receiving $n$ differentially private values from $n$ identified parties, the server now only receives a differentially private histogram representing the collective response of a (known) group of $n$ clients.
Clearly, in this situation, the client could answer more server queries within the same privacy budget, since the budget decreases less quickly (see \cref{sec:dp-algorithms}).
One can determine the influence on the privacy budget decrease for a particular randomizer in the shuffle model following, e.g.,~\cite{balle2019shuffle}.

An interesting question to ask is whether either of the previous schemes can be transformed to work in the shuffle model.
Whilst we assume that the shuffler properly randomizes all messages and does not provide the server with any other information, the messages themselves might still be linkable to the client who sent them.
In fact, we observe that neither the \pcalgostyle{Base} nor the \pcalgostyle{Expand} scheme can be applied directly in the shuffle model, since the public values $\pk_i, \sigma_x^{i,j}, \cm_{\rho_c}^{i,j}/\rt_i, k_s^{i,j}$, and $\sigma_s^{i,j}/\sigma_s^i$ are the same for different runs of $\pcalgostyle{Randomize}$.
This would allow the server to easily link several messages to the same client by simply comparing these public values.
Even worse, $\pk_i$ is directly linkable to the $i$-th client.

Fortunately, we can solve this, by moving these values to the witness part of the \nizkpk{} statement and include the verification statements on line 3 and 4 into $\relation_\text{shuffle}$.\footnote{The statement on line 2 is implicitly guaranteed by the check in \pcalgostyle{GenRand}.}
This transformation clearly guarantees unlinkability of different $\pcalgostyle{Randomize}$ messages of the same client.
Also, verifiable correctness is still guaranteed, which can be seen intuitively as follows.
First, observe that, since $\pk_i$ is included in $\sigma_s^i$, and we verify $\sigma_s^i$ inside the \nizkpk for $\relation_\text{shuffle}$, the client has to use a pair $(x_{i,j}, t_x^{i,j})$, signed by its own trusted environment.
Second, the inclusion of $\pk_i$ inside $\sigma_s^i$ also guarantees that the randomness is bound to a specific client, and thus a set of colluding clients could not interchange their random values.

In summary, this gives us the following high-level protocol execution. At time step $j$ each client sends a message containing a differentially private value $\tilde{x}_{i,j}$ and a proof $\pi_{i,j}$ to the shuffler.
Next, the shuffler collects all these messages and forwards them in random order to the server.
I.e., the server receives $n$ value-proof pairs. Finally, the server verifies the proof of each value, and accepts the values with a correct proof.
We note that for proof verification, the server only requires the corresponding $\tilde{x}$, public parameters \pp, it's own public key $\pk_s$ and the verification key \vk.
Since $\tilde{x}$ is the requested value and all other values are identical for all messages, we are certain that we do not compromise the unlinkability that is required in the shuffle model.

However, we do not only want a secure protocol.
The above construction still requires careful consideration to keep the client-side performance practical.
We note that by moving the verification of $\rho_s^{i,j} = \prf(k_s^i||s_j)$ to the \nizkpk, we can remove the Merkle tree.
This is done by having both the server and client $i$ generate a random value for the \prf seed, respectively $k_s^i$ and $k_c^i$.
The full $\prf$ seed is defined as $k_i = k_c^i \oplus k_s^i$.
Next, we compute a random value $\rho = \prf(k_i, s_j)$ and verify this inside the \nizkpk.
By doing this, we only require one verification of a $\prf$, rather than requiring both a $\prf$ verification and verifying the presence of a commitment in a Merkle tree.
We observe that we could also have used this construction in our \pcalgostyle{Expand} scheme, however the \nizkpk for practically sized Merkle trees is more efficient than that for a secure \prf evaluation~\cite{hopwoodZcashProtocolSpecification2023}.
A precise specification of the \pcalgostyle{Shuffle} scheme is given in \cref{fig:scheme-shuffle}.

\renewcommand{\pclnstyle}[1]{\hspace*{-3pt}\text{\fontsize{6}{7}\selectfont#1}}

\begin{figure}
\fontsize{6}{7}\selectfont
\begin{pcvstack}[boxed,center,space=0.1cm]

\procedure[codesize=\fontsize{6}{7}\selectfont]{$\pcalgostyle{VLDPPipeline}_\text{shuffle}$}{
    \pcln \pp \leftarrow \setup_\text{shuffle}(1^\lambda) \\
    \pcln \text{Server computes }(\ek, \vk, \pk_s, \sk_s, L) \leftarrow \kgen_\text{shuffle}(\pp) \\
    \pcln \pcfor \text{Each client $i$ (in parallel)} \\
    \pcln \t \text{Client obtains } \out_c^i= \pcalgostyle{GenRand}_\text{shuffle}(\pp) \\
    \pcln \t \pcfor j \pcin~\{1,\ldots,T\} \\
    \pcln \t \t \text{Client $i$ obtains fresh } (x_{i,j}, t_x^{i,j}, \sigma_x^{i,j}=\sig.\sign_{\sk_i}(x_{i,j}||t_x^{i,j})) \\
    \pcln \t \t \text{Client $i$ runs } (\tilde{x}_{i,j}, \pi_{i,j}) = \pcalgostyle{Randomize}_\text{shuffle}(\pp, \ek, t_j, \out_c^i, x^{i,j}, t_x^{i,j}, \sigma_x^{i,j}) \\
    \pcln \t \t \text{Shuffler forwards messages in random order} \\
    \pcln \t \t \text{Server obtains } \tilde{x}_{?_i,j} = \pcalgostyle{Verify}_\text{shuffle}(\pp, \vk, t_j, \tilde{x}^{?_i,j}, \pi_{?_i,j}) \\
    \pcln \text{Server computes result from all $\tilde{x}_{?_i,j}$}
}

\begin{pchstack}[space=0.1cm]
\begin{pcvstack}[space=0.1cm]
\procedure[linenumbering,codesize=\fontsize{6}{7}\selectfont]{$\kgen_\text{shuffle}(\pp)$}{
    (\ek, \vk) \leftarrow \nizkpk.\kgen(\relation_\text{shuffle}) \\
    (\sk_s, \pk_s) \leftarrow \sig.\kgen(\pp_\text{sig}) \\
    L \leftarrow \emptyset \\
    \pcreturn (\ek, \vk, \pk_s, \sk_s, L)
}

\procedure[linenumbering,codesize=\fontsize{6}{7}\selectfont]{$\setup_\text{shuffle}(1^\secpar)$}{
    \pp_\text{sig} \leftarrow \sig.\setup(1^\secpar) \\
    \pp_\text{comm} \leftarrow \comm.\setup(1^\secpar) \\
    \vec{s} = (s_1, \ldots, s_T) \sample \bin^{\lambda \times T}\\
    \vec{t} = (t_0, \ldots, t_T) \\
    \pp = (\relation_\text{shuffle}, \pp_\text{sig}, \pp_\text{comm}, \vec{s}, \vec{t}) \\
    \pcreturn \pp
}
\end{pcvstack}

\procedure[codesize=\fontsize{6}{7}\selectfont]{$\relation_\text{shuffle}$}{
    \text{Given $(t_{j-1}, t_j, \pk_s, s_j, \tilde{x}_{i,j})$,} \\
    \text{the prover knows} (t_x^{i,j}, x^{i,j}, \pk_i, \sigma_x^{i,j}, k_c^i, \\
    \t r_{k_c}^i, \cm_{k_c}^i, k_s^i, \sigma_s^i) \text{ such that:} \\
    \pcln t_x^{i,j} \in (t_{j-1}, t_j] \\
    \pcln \sig.\verify_{\pk_i}(\sigma_x^{i,j},x_{i,j}||t_x^{i,j}) = 1 \\
    \pcln \cm_{k_c}^i = \comm(k_c^i; r_{k_c}^i) \\
    \pcln k_i = k_c^i \oplus k_s^i \\
    \pcln \sig.\verify_{\pk_s}(\sigma_s^i,\pk_i || \cm_{k_c}^i || k_s^i) = 1 \\
    \pcln \rho_{i,j} = \prf(k_i, s_j) \\
    \pcln \tilde{x}_{i,j} = \ldp.\apply(x_{i,j}; \rho_{i,j})
}
\end{pchstack}

\begin{pchstack}[space=0.1cm]
\procedure[codesize=\fontsize{6}{7}\selectfont]{$\pcalgostyle{GenRand}_\text{shuffle}(\pp)$ --- \textbf{Client $i$}}{
    \pcln k_c^i \sample \bin^* \> \> \\
    \pcln r_{k_c}^i \sample \bin^* \> \> \\
    \pcln \cm_{k_c}^i = \comm(k_c^i; r_{k_c}^i) \\
    \pcln \text{Send $(\pk_i, \cm_{k_c}^i)$ to server} \\
    \pcln \text{Receive $(k_s^i, \sigma_s^i)$ from server} \\
    \pcln \text{If } \sig.\verify_{\pk_s}(\sigma_s^i,\pk_i || \cm_{k_c} || k_s^i) \neq 1,\text{ abort} \\
    \pcln \pcreturn \out_c^i = (k_c^i, r_{k_c}^i, \cm_{k_c}^i, k_s^i, \sigma_s^i)
}

\procedure[codesize=\fontsize{6}{7}\selectfont]{$\pcalgostyle{GenRand}_\text{shuffle}(\pp)$ --- \textbf{Server}}{
    \pcln \text{Receive $(\pk_i, \cm_{k_c}^i)$ from client $i$} \\
    \pcln \text{If $\pk_i$ does not belong to $i$, abort} \\
    \pcln \text{If $i \in L$, abort} \\
    \pcln L \leftarrow L \cup \{i\} \\
    \pcln k_s^i \sample \bin^* \\
    \pcln \sigma_s^i = \sig.\sign_{\sk_s}(\pk_i || \cm_{k_c}^i || k_s^i) \\
    \pcln \text{Send $(k_s^i, \sigma_s^i)$ to client $i$}
}
\end{pchstack}

\begin{pchstack}[space=0.1cm]
\procedure[codesize=\fontsize{6}{7}\selectfont]{$\pcalgostyle{Randomize}_\text{shuffle}(\pp, \ek, t_j, \out_c^i, x_{i,j}, t_x^{i,j}, \sigma_x^{i,j})$}{
    \pcln k_i = k_c^i \oplus k_s^i \\
    \pcln \rho_{i,j} = \prf(k_i, s_j) \\
    \pcln \tilde{x}_{i,j} = \ldp.\apply(x; \rho) \\
    \pcln \vec{\phi}_{i,j} = (t_{j-1}, t_j, \pk_s, s_j, \tilde{x}_{i,j}) \\
    \pcln \vec{w}_{i,j} = (t_x^{i,j}, x_{i,j}, \pk_i, \sigma_x^{i,j}, k_c^i,r_{k_c}^i, \cm_{k_c}^i, k_s^i, \sigma_s^i) \\
    \pcln \pi = \nizkpk_\text{shuffle}.\prove(\vec{\phi}_{i,j}; \vec{w}_{i,j}) \\
    \pcln \text{ Send $(\pi_{i,j}, \tilde{x}_{i,j})$ to shuffler}
}
\procedure[codesize=\fontsize{6}{7}\selectfont]{$\pcalgostyle{Verify}_\text{shuffle}(\pp, \vk, t_j, \tilde{x}_{i,j}, \pi_{i,j})$}{
    \pcln \vec{\phi} = (t_{j-1}, t_j, \pk_s, s_j, \tilde{x}_{i,j}) \\
    \pcln \text{If } \nizkpk.\verify( \\
    \t \t \t \pi^i, \tilde{x}_{i,j},\pk_s, s_j) \neq 1 \text{, abort} \\
    \pcln \pcreturn \tilde{x}_{i,j}
}
\end{pchstack}

\end{pcvstack}
\caption{\pcalgostyle{Shuffle} scheme: efficient VLDP in the shuffle model.}
\label{fig:scheme-shuffle}
\Description{Details of the \pcalgostyle{Shuffle} scheme.}
\end{figure}

\renewcommand{\pclnstyle}[1]{\text{\scriptsize#1}}

\section{Experimental Evaluation}
\label{sec:evaluation}
To assess the practical performance of our constructions and to compare different versions, we conducted various experiments on synthetic and real data, and report the communication costs and computation times.
We first describe our implementation of the schemes, including how the different building blocks were instantiated.
This is followed by a description of our experiments and their results to support our efficiency and practicality claims.

\subsection{Implementation}\label{subsec:implementation}
Each scheme was implemented in Rust using the Arkworks v0.4 library~\cite{arkworks2022arkworks}.
It provides efficient implementations for zk-SNARKs and other cryptographic primitives with gadgets to evaluate these primitives inside a zk-SNARK circuit.
Our cryptographic building blocks are instantiated as follows, targeting 128-bit security:
\begin{itemize}[leftmargin=*]
\item \nizkpk: The \emph{Groth16} zk-SNARK~\cite{grothSizePairingBasedNoninteractive2016} is used to generate the NIZK-PKs.
This specific pairing-based, circuit zk-SNARK scheme has gained widespread adoption in real-world applications due to its efficiency and constant proof size.
This scheme does rely on a trusted setup, which, if broken, would allow anyone to create false proofs.
However, this is not an issue in our constructions, since the server can execute this trusted setup by itself.
Furthermore, the server is assumed to behave semi-honestly and to be non-colluding.
The zk-SNARK elements are chosen to be on the \emph{BLS12-381} elliptic curve (EC)~\cite{boweBLS12381NewZkSNARK2017}, which is a known pairing-friendly curve with good (estimated 128-bit) security.
Moreover, there is a known embedded curve for BLS12-381, called \emph{Jubjub}~\cite{boweZkcryptoJubjub2024}, which allows for efficient and secure evaluation of EC-primitives inside zk-SNARK circuits.
\item \sig: The signature schemes used by client and server are both implemented using Schnorr signatures~\cite{schnorrEfficientIdentificationSignatures1990}.
Specifically, we use EC-Schnorr signatures over the Jubjub curve, due to its efficient verification inside a zk-SNARK circuit~\cite{steidtmannBenchmarkingZKCircuitsCircom2023}.
Moreover, this scheme is often used in practice, and other popular schemes, such as EdDSA (${\sim}$1,000 more constraints) and ECDSA (${\sim}$10,000 more) would only have a small to negligible performance impact~\cite{steidtmannBenchmarkingZKCircuitsCircom2023}.
Additionally, we use the Blake2s-256 collision resistant hash (CRH)~\cite{saarinenBLAKE2CryptographicHash2015} to hash the input message to a fixed length digest.
This CRH was chosen for its good security (128 bits against collision attacks), and efficiency inside a zk-SNARK\@.
\item $\comm$: Our commitment scheme is instantiated using Pedersen vector commitments~\cite{pedersenNonInteractiveInformationTheoreticSecure1992} (with 4-bit windows) over the Jubjub curve.
This instantiation is very efficient inside a zk-SNARK circuit, is information-theoretically hiding, and targets the required bit security for the binding property.
\item \prf: We construct a \prf using keyed Blake2s-256~\cite{saarinenBLAKE2CryptographicHash2015}, which gives a \prf output of 256 bits, or 32 bytes.
Also here, Blake2s was chosen to fit the targeted security level, whilst still being practical inside a zk-SNARK circuit.
\item \merkletree: This primitive is only used inside the \textsf{Expand} scheme.
By using Pedersen commitments to instantiate $\comm$, we can use these commitments directly as the leaves of the Merkle tree due to their fixed size (which is no more than 256 bits in our case).
To compute the higher level nodes and root, we use the Pedersen hash function~\cite{hopwoodZcashProtocolSpecification2023} to hash the concatenation of both its children.
We use a Pedersen hash rather than Blake2s here, since it is significantly more efficient inside a zk-SNARK circuit, and has security guarantees similar to that of Groth16, thereby not decreasing the security of our scheme.
Finally, we note that the tree depth $d$ has to be the smallest power of 2 such that $2^{d-1} \geq T$, where $T$ is the total number of time steps we wish to run.
\end{itemize}

Our open-source code\footnote{The source code of our implementation is available at \url{https://github.com/xQiratNL/VLDP}.} was written in such a way that it is simple to replace a specific building block by another.
In \cref{subsec:optimizations-alternatives}, we discuss alternative building block choices and their impact on security, efficiency, and practicality.

\begin{table*}
\centering
\scriptsize
\begin{tabular}{@{}llrrrrrrrrrrr@{}}
\toprule
  \multirow{2}{*}{Dataset} &
  \multirow{2}{*}{Scheme} &
  \multicolumn{3}{c}{Client} &
  \multicolumn{2}{c}{Server} &
  \multicolumn{3}{c}{Communication} & 
  \multicolumn{3}{c}{\nizkpk} \\ \cmidrule(lr){3-5} \cmidrule(lr){6-7} \cmidrule(lr){8-10} \cmidrule(l){11-13}
  &
  &
  \multicolumn{1}{c}{GenRand-1} &
  \multicolumn{1}{c}{GenRand-2} &
  \multicolumn{1}{c}{Randomize} &
  \multicolumn{1}{c}{GenRand} &
  \multicolumn{1}{c}{Verify} &
  \multicolumn{1}{c}{GenRand-1} &
  \multicolumn{1}{c}{GenRand-2} &
  \multicolumn{1}{c}{Randomize} &
  \multicolumn{1}{c}{$|\ek|$} &
  \multicolumn{1}{c}{$|\vk|$} &
  \multicolumn{1}{c}{\# constraints} \\ \midrule
  \multirowcell{3}{Geolife \\ GPS} & \textsf{Base}    & 0.218 ms & 0.302 ms & 0.610 s & 2.494 ms & 3.454 ms & 65 B & 96 B & 360 B & 16.1 MB & 776 B & 55 884 \\
                                   & \textsf{Expand}  & 3.033 ms & 0.267 ms & 1.213 s & 2.005 ms & 4.425 ms & 64 B & 96 B & 360 B & 23.4 MB & 824 B & 74 322 \\
                                   & \textsf{Shuffle} & 0.230 ms & 0.305 ms & 1.798 s & 2.413 ms & 2.680 ms & 64 B & 96 B & 200 B & 53.2 MB & 728 B & 173 460\\ \addlinespace[1em]
  \multirowcell{3}{Smart \\ meter} & \textsf{Base}    & 0.223 ms & 0.213 ms & 0.619 s & 2.146 ms & 3.484 ms & 65 B & 96 B & 360 B & 16.3 MB & 776 B & 56 903 \\
                                   & \textsf{Expand}  & 2.475 ms & 0.211 ms & 1.106 s & 1.271 ms & 3.540 ms & 64 B & 96 B & 360 B & 23.7 MB & 824 B & 75 341 \\
                                   & \textsf{Shuffle} & 0.225 ms & 0.293 ms & 1.821 s & 2.119 ms & 2.659 ms & 64 B & 96 B & 200 B & 53.3 MB & 728 B & 174 095\\ \bottomrule \addlinespace
\end{tabular}
\caption{Performance of all schemes for two real-world applications: client and server computation and communication costs of a \underline{single} evaluation of an algorithm/protocol, byte size of $\ek$ and $\vk$, and number of constraints in the \nizkpk.}
\label{tab:metrics_single}
\end{table*}
\begin{table}
    \centering
    \scriptsize
    \begin{tabular}{@{}llrrrrr@{}}
        \toprule
        \multirow{2}{*}{Dataset} &
        \multirow{2}{*}{Scheme} &
        \multicolumn{3}{c}{Client} &
        \multicolumn{2}{c}{Server} \\ \cmidrule(lr){3-5} \cmidrule(l){6-7}
        &
        &
        \multicolumn{1}{c}{GenRand-1} &
        \multicolumn{1}{c}{GenRand-2} &
        \multicolumn{1}{c}{Rand.} &
        \multicolumn{1}{c}{GenRand} &
        \multicolumn{1}{c}{Verify} \\ \midrule
        \multirowcell{3}{Geolife\\GPS} & \textsf{Base}    & 1.092 ms & 1.508 ms & 3.032 s &  2.392 s &  3.316 s \\
        & \textsf{Expand}  & 3.033 ms & 0.267 ms & 6.045 s &  0.385 s &  4.425 s \\
        & \textsf{Shuffle} & 0.230 ms & 0.305 ms & 8.973 s &  0.463 s &  2.572 s \\ \addlinespace[1em]
        \multirowcell{3}{Smart\\meter} & \textsf{Base}    & 1.113 ms & 1.064 ms & 3.078 s & 59.734 s & 96.977 s \\
        & \textsf{Expand}  & 2.475 ms & 0.211 ms & 5.510 s &  7.076 s & 98.536 s \\
        & \textsf{Shuffle} & 0.225 ms & 0.293 ms & 9.090 s & 11.796 s & 74.999 s \\ \bottomrule \addlinespace
    \end{tabular}
    \caption{\underline{Total} computation time for both use cases, over all time steps $(T=5)$ (for the server also over all clients).}
    \label{tab:runtime_total}
\end{table}

\subsection{Experimental Setup}\label{subsec:experiment-setup}
We perform two sets of experiments to evaluate and compare the practical performance of our constructions.
The first set uses two real datasets to evaluate and validate the performance of each scheme in a real-world setting.
The second set of experiments uses synthetic data to evaluate the scalability of our schemes.

We use two datasets in our experiments, and consider $T=5$ days of readings from each.
The first dataset (\emph{Geolife GPS Trajectory}) is a location dataset of 182 users, which after pre-processing, gave us 8 potential postcode locations per day for the subject group.
With respect to the algorithm for histograms (see \cref{fig:ldp-algos}), we thus have $k=8$, where $\{1,\ldots,k\}$ represent the respective postcodes.

The second dataset (\emph{Smart meter}) contains smart meter energy readings (floating points) of 5,567 households.
We use a precision level of $k = 10$ (see the algorithm for reals in \cref{fig:ldp-algos}) for our experiments.
Both datasets are described in more detail in \cref{subsec:dataset-description}.

\subsubsection*{Experiments}
For our experiments, we determine the median runtime of 100 runs (after discarding three warm-up runs), of each of the algorithms at the client and server side.
Specifically, we look at the computation time for individual clients and the server in the different phases.
Next to this, we also measure the byte size of all (compressed) messages.
The experiments were run on a desktop computer with Windows 10 desktop PC with a Ryzen 3600 CPU with 6 cores and 12 threads @4.0GHz and 16GB dual-channel DDR4 RAM at 3600MHz.
The experiments were run using Rust 1.77.2.

\subsection{Concrete Applications}
\label{subsec:conc-apps}
For both datasets, the timestamp is encoded using one byte.
Next to this, we use 8 bytes (64 bits) for each random value we sample.
This guarantees statistical closeness to the true distributions, since $k \ll 2^{64}$.
Thus, for the Geolife GPS dataset, i.e., histogram, we require 16 bytes of randomness ($|\rho| = 16$). For the smart meter dataset, i.e., real valued data, we require 24 bytes of randomness ($|\rho| = 24$), $2 \cdot 8$ bytes for both Bernoulli samples and another 8 bytes for one uniformly random sample.
Both are below the 32 bytes that we get as output from one \prf evaluation.
In \cref{subsec:general-performance}, we evaluate the computation times and message sizes for larger values of $|\rho|$.
The performance of our schemes is not impacted by any particular value of $\epsilon$ and $\delta$ used in the DP mechanism.
For completeness, we shall use 5 runs of the LDP mechanism, with the privacy budget per run, i.e., $\epsilon_0$, determined as in \cref{sec:dp-algorithms}.
Finally, for the \textsf{Expand} scheme we set the Merkle tree depth to $d_\textsf{MT}=4$, i.e., it has $2^{4-1} = 8$ leaves, since we run both datasets for $T=5$ steps.

\subsubsection*{Results}
The median computation times and message sizes for a single run of each algorithm are shown in \cref{tab:metrics_single}.
This table also includes the size of the \nizkpk{} evaluation/verification key, and the number of constraints.
The evaluation key is relatively large, and needs to be communicated with each client.
Fortunately, its generation is part of the setup and can be communicated as part of the public parameters beforehand.
The number of constraints gives an implementation-independent view on the proof generation and verification costs and is the most fair way to compare different schemes.
Especially, since proof generation and verification are the dominating factors in $\pcalgostyle{Randomize}$ and $\pcalgostyle{Verify}$.
To better understand the cost in both use cases, we report the overall computation time for the server and of a single client in \cref{tab:runtime_total}.

Regarding the communication costs of $\textsf{Base}$, we see that each client sends $(65+360)T=2{,}125$ bytes (B) to the server, and receives $96T=480\text{B}$.
In the \textsf{Expand} scheme, this reduces to $64+360T=1{,}864$ sent and $96$ received bytes.
For the $\textsf{Shuffle}$ scheme, the amount of bytes sent by each client reduces further to only $64+200T=1{,}064\text{B}$.

Moreover, we observe that the $\textsf{Base}$ scheme puts a much higher load on the server, in both computation and communication\footnote{Since $\pcalgostyle{GenRand}$ has to be run once per time step, instead of once overall, $\pcalgostyle{Base}$ has $T\times$ more communication than other schemes for $\pcalgostyle{GenRand}$.} costs, in the $\pcalgostyle{GenRand}$ phase.
This is due to the fact that this phase needs to be run again for each time step.
$\textsf{Expand}$ requires slightly more computational effort from each client, however, this is negligible when compared to the reduction in server computation time, and overall communication costs.
The $\textsf{Shuffle}$ scheme requires the least effort in this phase, but puts clearly higher cost on the client in $\pcalgostyle{Verify}$.
It should be noted, however, that the computational cost for the client is very practical and lies in the 0.5--2 seconds range for all schemes.
Moreover, the computation and communication costs of $\pcalgostyle{Verify}$ are significantly lower for $\textsf{Shuffle}$, which makes it more attractive even in the `regular' local model. 
We remark that we did not implement server-side parallelization.
Hence, the server's runtime (\cref{tab:runtime_total}) could be further reduced by, e.g., distributing the client messages over different processes.

Additionally, we note that the shuffle model does introduce some additional latency, when compared to the `regular' local model. The amount of latency depends on how the shuffler is implemented (see \cref{sec:implement-shuffle}). But, given that the message size is small, we expect this to be of little influence in most applications.
The introduced latency will be in the same order as in the shuffle model without verifiability, i.e., our introduction of verifiability does not in itself restrict the usage of a shuffler.

\subsection{General Performance \& Comparison}\label{subsec:general-performance}
Above, we evaluated our schemes on common LDP algorithms and datasets, thereby reflecting the expected application settings for our protocols.
As discussed, the amount of randomness required ($|\rho|$) for $\ldp.\apply()$ determines the majority of the computation cost, due to the cost of evaluating \prf.
Here, $|\rho|$ depends on two factors:
\begin{itemize}[leftmargin=*]
    \item \emph{The randomizer:} The amount of random variables used, and the number of bits to (accurately) sample them, predominantly determines the size of $|\rho|$. In other words, more random values or sampling with higher accuracy leads to an increase in $|\rho|$.
    \item \emph{Entropy of released data:} The data that is released by the client also has an effect on $|\rho|$, albeit indirectly. Namely, when releasing data with higher entropy (more records, larger domain size) more randomness is needed to ensure differential privacy. For example, privately releasing a bit requires less randomness than releasing an 8-bit integer. Similarly, releasing 10 values requires more randomness than only 1.
    Thus, when considering higher-dimensional datasets one often also releases data with higher entropy and thus indirectly requires more randomness.
    Finally, we note that using higher-dimensional input data could lead to reduced performance due to increasing the signature input size.
    Fortunately, this can be counteracted by using SNARK-friendly signatures (see \cref{subsec:optimizations-alternatives}).
\end{itemize}

To better investigate the performance impact of the amount of randomness used, we vary $|\rho|$ in steps of 32, which is the output size of our choice of \prf, i.e., smaller step sizes will show negligible differences in performance.
This gives insight into the performance of other LDP mechanisms as discussed in \cref{subsec:ldp-inside-nizk}.
Next to this, we investigate the performance of the \textsf{Expand} scheme for different Merkle tree depths.
We vary $d_\textsf{MT}$ between 2 and 11, i.e., from 2 to 1,024 leaves, which should be more than sufficient in realistic settings.
In all experiments, we use randomly generated data, and encode the timestamps and input values as 64-bit values.

\begin{figure*}
    \centering
    \includegraphics[width=0.33\textwidth]{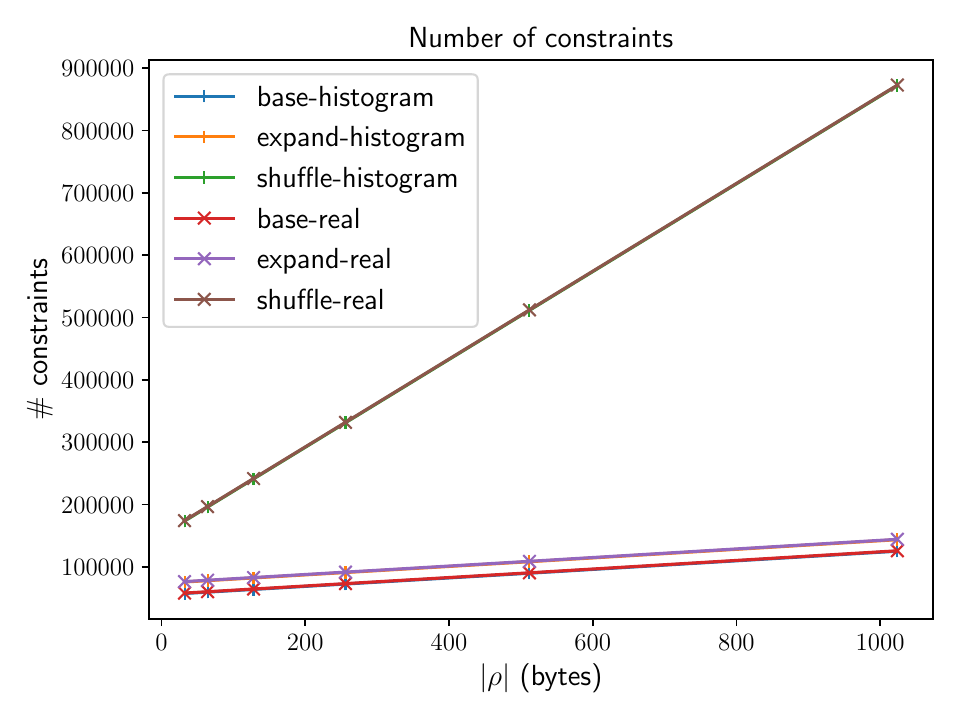}
    \includegraphics[width=0.33\textwidth]{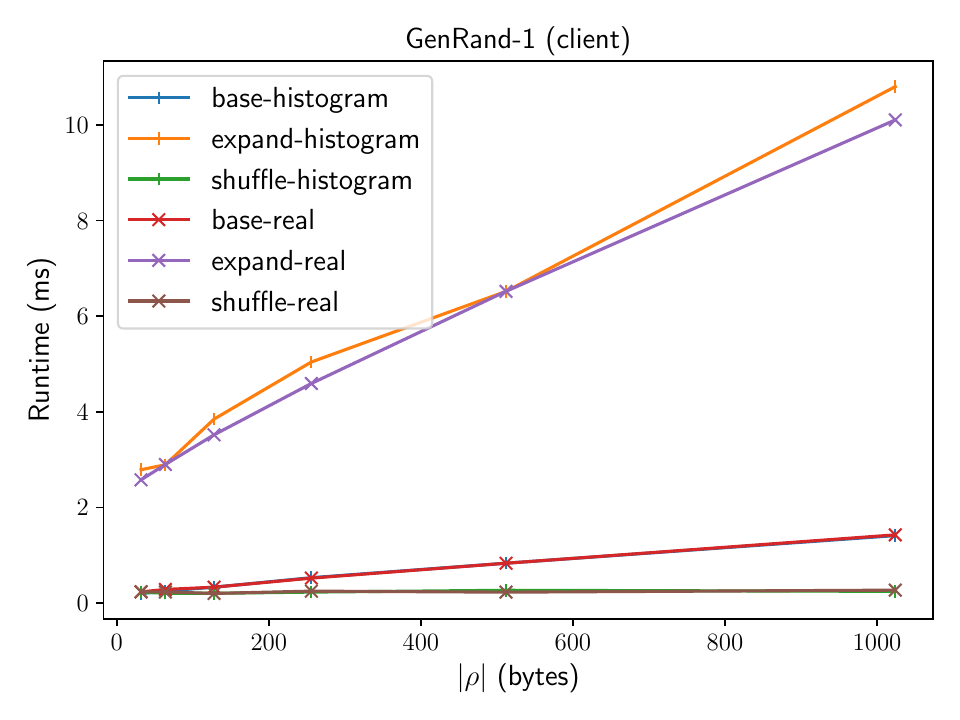}
    \includegraphics[width=0.33\textwidth]{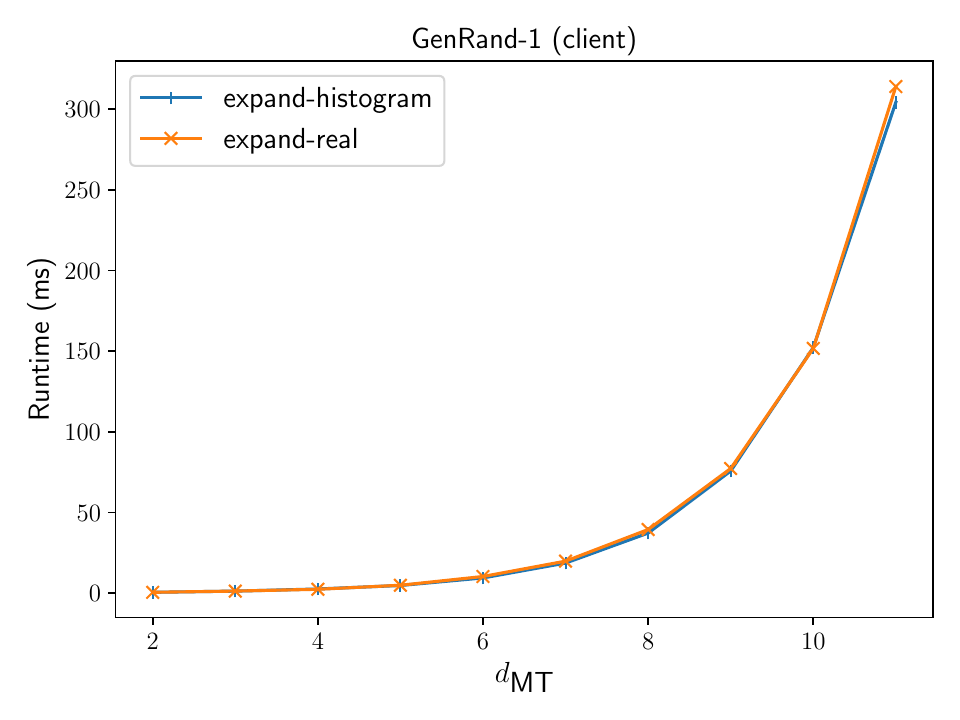}
    \includegraphics[width=0.33\textwidth]{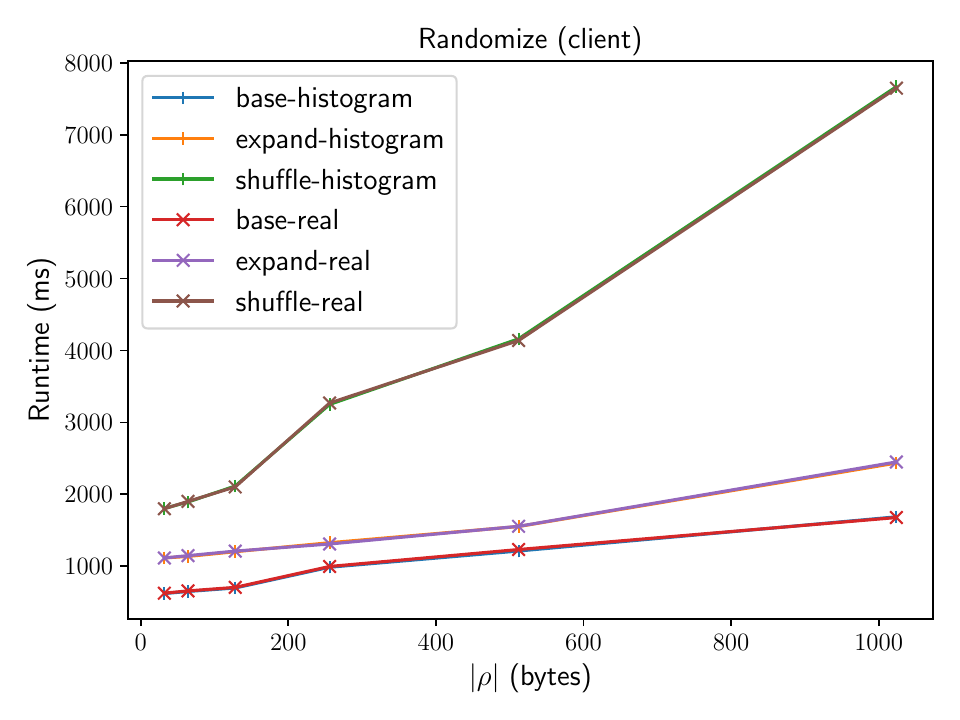}
    \includegraphics[width=0.33\textwidth]{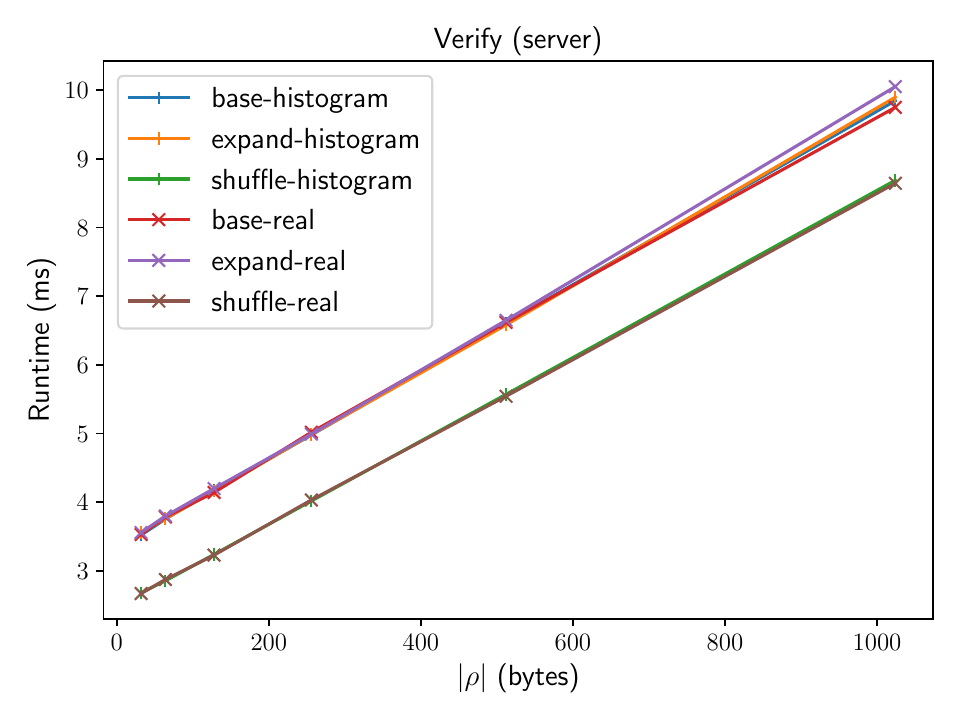}
    \includegraphics[width=0.33\textwidth]{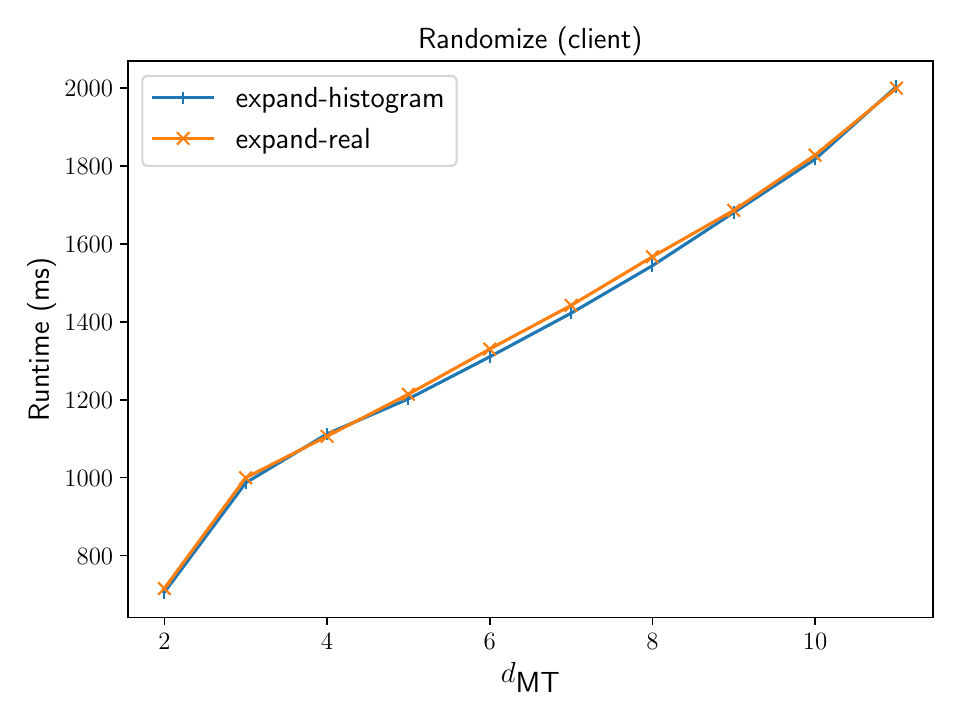}
    \caption{Impact of $|\rho|$ on \# constraints (topleft), client runtime in \pcalgostyle{GenerateRandomness} (topcenter) \& \pcalgostyle{Randomize} (bottomcenter), and server runtime in \pcalgostyle{Verify} (bottomleft); Impact of $d_\textsf{MT}$ on client runtime in \pcalgostyle{GenRand} (topright) and \pcalgostyle{Randomize} (bottomright).}
    \label{fig:graphs-rand}
        \Description{Impact of $|\rho|$ and $d_\textsf{MT}$ on the runtime of the different schemes.}
\end{figure*}

\subsubsection*{Results}
First, we observe that the communication size is independent of both $|\rho|$ and $d_\textsf{MT}$, and thus only look at their influence on the runtime (see \cref{fig:graphs-rand}).
An increase of $|\rho|$ leads to a significant increase in the constraint count and duration of $\pcalgostyle{Randomize}$ for the $\textsf{Shuffle}$ scheme.
However, even for as much as 1,024 bytes of randomness, the computation time remains below 8 seconds, which is still very practical.
As shown in \cref{subsec:ldp-inside-nizk}, 1,024 is more than sufficient for 64-bit precision in many representative and state-of-the-art LDP algorithms.
Additionally, we observe a significant, but approximately linear, increase in computation time for the client in $\pcalgostyle{GenRand}$.
Since the duration is in the millisecond range, this will not cause any practical issues.
Finally, we see a linear increase in the verification time for the server.
However, this verification time is so small that we do not consider it to be an issue.
For the $\pcalgostyle{Expand}$ scheme, we observe a linear increase of the number of constraints (from ${\sim}60,000$ to ${\sim}120,000$) and an exponential increase of the duration of $\pcalgostyle{GenRand}$ in $d_\textsf{MT}$ (\cref{fig:graphs-rand}).
This agrees with the fact that the amount of random values grows exponentially in $d_\textsf{MT}$.
The increase in runtime for $\textsf{Randomize}$ is also approximately linear, and only takes around 2 seconds for a Merkle tree with 1,024 random values.

\subsubsection*{Comparison}
In conclusion, we see that the total runtime of each scheme scales approximately linearly in the amount of randomness required, for both client and server, and that the runtime of each schemes is very practical for realistically sized parameters.
Clearly, the \textsf{Shuffle} scheme has the lowest communication cost and server load, in addition to being secure in the shuffle model.
Conversely, $\textsf{Expand}$ puts a smaller load on the client, and slightly higher on the server, but is not secure in the shuffle model.
Finally, the $\textsf{Base}$ scheme puts a comparatively high communication and computation load on the server, making it less practical than the other schemes.

For comparison with related work, we consider~\cite{movsowitzdavidowPrivacyPreservingTransactionsVerifiable2023} which is the work closest to our construction.
As mentioned, the number of constraints provides a fair comparison for different schemes.
Their scheme requires two \nizkpk proofs for one transfer of LDP values.
For one input their proofs have 9,769 and 12,882 constraints, i.e., the combined number of constraints is around 2.5---7.5 times smaller than our scheme, which means the computational effort for both client and server will also be smaller by a similar factor.
However, the underlying blockchain structure used in~\cite{movsowitzdavidowPrivacyPreservingTransactionsVerifiable2023} will also come with its own latency and scalability issues, which our scheme does not suffer from.
On top of that, it is only evaluated for binary RR, which is much simpler than our construction.
Finally,~\cite{movsowitzdavidowPrivacyPreservingTransactionsVerifiable2023} does not discuss the performance of their approach to \pcalgostyle{GenRand}.
However, as it is similar in nature to $\pcalgostyle{GenRand}_\text{base}$, it will suffer from the same drawbacks when compared to \textsf{Expand} and \textsf{Shuffle}.

\section{Conclusion}\label{sec:conclusion}
We showed how to construct verifiable LDP schemes for both the local and, most interestingly, the shuffle model, which guarantee security against data manipulation attacks.
Experimental evaluation of our schemes on realistic use cases underscores their practicality.
Especially the \textsf{Expand} and \textsf{Shuffle} schemes put a very low load (5--7 ms) on the server, whilst keeping client computation times down to $<2$ seconds.
Moreover, we showed the scalability of our schemes using generic benchmarks.
Finally, we discuss how our schemes can be efficiently adopt a wide variety of LDP algorithms, due to their generic design.

\begin{acks}
We thank the anonymous reviewers for their insightful feedback, which helped improve our work.
This research received no specific grant from any funding agency in the public, commercial, or not-for-profit sectors.
\end{acks}

\bibliographystyle{ACM-Reference-Format}
\bibliography{main}

\appendix

\section{Appendix to Section 4}
\label[appendix]{sec:appendix-to-section-4}

\subsection{De-Biased Output}
\label[appendix]{subsec:de-bias}
Let $X_i$ denote the random variable representing user $i$'s output after running the LDP algorithm for reals.
Let $X = \sum_i^n X_i$.
We are interested in finding:
\[
    \mathbb{E}(X) = \sum_{i=1}^n \mathbb{E}(X_i)
\]
Let $p_j$ be the probability that user $i$ outputs $j \in \{0, 1, \ldots, k\}$.
Let $q_j$ be the true probability of any user having input $j$.
Then,
\begin{align*}
    p_j &= \left( 1 - \gamma + \frac{\gamma}{k+1}\right) q_j + \frac{\gamma}{k+1} (1 - q_j)\\
    &= (1 - \gamma)q_j + \frac{\gamma}{k+1}
\end{align*}
Then
\begin{align*}
    \mathbb{E}(X_i) &= \sum_{j = 0}^k j p_j \\
    &=\sum_{j = 0}^k j \left( (1 - \gamma)q_j + \frac{\gamma}{k+1} \right)\\
    &= (1 - \gamma) \left( \sum_{j = 0}^k j q_j \right) + \frac{\gamma k}{2} \\
    &= (1 - \gamma) \mu + \frac{\gamma k}{2}
\end{align*}
where $\mu = \sum_{j = 0}^k j q_j$ is the true expected input of any user.
Thus,
\begin{align*}
    \mathbb{E}(X) &= n \left( (1 - \gamma) \mu + \frac{\gamma k}{2}\right)\\
    \Rightarrow \frac{n\mu}{k} &= \frac{1}{1 - \gamma}\left( \frac{\mathbb{E}(X)}{k} - \frac{\gamma n}{2} \right).
\end{align*}
Therefore, the expected value of the sum to precision $k$ output by the LDP algorithm, i.e., $\mathbb{E}(X)/k$, gives us the expectation of the sum of true inputs to precision $k$, i.e., $n\mu/k$.
Thus, given the sum of these values for a sample, we can estimate the true sum as above.

\subsection{Additional, Simple Example}
\label[appendix]{subsec:example-ldp-nizk}
To better explain the requirements above, consider the example where we sample a random element from $\{0,1,2\}$ using uniform random bits.
In practice, an often used method to achieve this is to sample two random bits, and map this as follows $00 \rightarrow 0; 01 \rightarrow 1; 10 \rightarrow 2$.
If the random bits are $11$ we sample two new random bits and repeat the process, until we terminate.\footnote{While the actual method might vary in practice, this simple version that we present here, is sufficient to describe the problem in the context of \nizk proofs.}
This process terminates (with probability 1), after a finite number steps.
However, due the variable requirement of random bits, we cannot use this sampling method inside a \nizk proof.
When considering this more closely, it becomes evident that there is no way to sample a random element from $\{0,1,2\}$ using a fixed number of random bits.
This problem occurs in many random sampling problems, but can fortunately easily be solved, by sampling from an approximate distribution that is statistically close to the true distribution.

For this example, we can sample an $\ell$-bit number $\rho$, such that $2^\ell$ is sufficiently large.
Then, we determine 3 intervals: $[0, \lfloor2^\ell/3\rfloor)$, $[2^\ell/3, 2\cdot \lfloor2^\ell/3\rfloor)$, and $[2\cdot \lfloor2^\ell/3\rfloor, 2^\ell - 1]$.
If $\rho$ is part of the $j$-th interval, we return $j$ as our random sample.
Observe, that all but the last interval have the exact same size of $\lfloor2^\ell/3 \rfloor$, and only the final interval contains $2^\ell -3\cdot \lfloor2^\ell/3\rfloor \leq 2$ more elements.
Thus, for sufficiently large $\ell$, the distribution generated by this sampling method is statistically close to the true distribution we wish to sample from.

\subsection{Approximate LDP Randomizers}\label[appendix]{subsec:approximate-ldp-randomizers}
In \cref{fig:approximate-ldp}, we give the precise specification of our approximate LDP randomizers for histograms and reals, based on the algorithms from \cref{sec:dp-algorithms}.
\begin{figure}[H]
    \begin{pchstack}[boxed,center]
        \scriptsize
        \pseudocode[
        head={$\ldp.\apply(x; \rho)$ for Reals},
        codesize=\scriptsize
        ]{
            \text{\textbf{input: }} k \in \mathbb{N}, \gamma \in [0, 1], \\
            \t \t \t \t \t x \in [0, 1], \rho \in \bin^*\\
            \text{Split $\rho$ into $(\rho_1, \rho_2, \rho_3)$}\\
            \overline{x} \gets \lfloor xk \rfloor + \widetilde{\text{Ber}}(xk - \lfloor xk \rfloor; \rho_1)\\
            b \gets \widetilde{\text{Ber}}(\gamma; \rho_2) \\
            \pcif b = 0 \pcdo \\
            \t \tilde{x} \gets \overline{x}\\
            \pcelse \\
            \t \tilde{x} \gets \widetilde{\text{Unif}}([0, k]; \rho_3)\\
            \pcreturn \tilde{x}
        }
        \;\;
        \pseudocode[
        head={$\ldp.\apply(x; \rho)$ for Histograms},
        codesize=\scriptsize
        ]{
            \text{\textbf{input: }} k \in \mathbb{N}, \gamma \in [0, 1], \\
            \t \t \t \t \t x \in [k], \rho \in \bin^*\\
            \text{Split $\rho$ into $(\rho_1, \rho_2)$}\\
            b \gets \widetilde{\text{Ber}}(\gamma; \rho_1) \\
            \pcif b = 0 \pcdo \\
            \t \tilde{x} \gets \overline{x}\\
            \pcelse \\
            \t \tilde{x} \gets \widetilde{\text{Unif}}([1, k]; \rho_2)\\
            \pcreturn \tilde{x}
        }
    \end{pchstack}
    \Description{Details of the approximate LDP algorithms for reals and histograms.}
    \caption{Approximate randomizers for reals and histograms.}
    \label{fig:approximate-ldp}
\end{figure}

\section{Implementing the Shuffler}
\label[appendix]{sec:implement-shuffle}
In practice a trusted shuffler can be implemented in a number of ways.
One way to do this is by using a \emph{mixnet}, or mix network.
A mixnet is a network involving several parties, that takes as input a list of messages and returns the same messages in a randomly permuted order.
Mixnets were first introduced in~\cite{chaumUntraceableElectronicMail1981} to realize untraceable e-mail and can be implemented in a variety of ways.
An overview can be found in, e.g.,~\cite{sampigethayaSurveyMixNetworks2006}.
In its most basic form, mixnets are implemented using a publicly known sequence of servers, whose public encryption keys are also available.
Any client wishing to send a message, encrypts their message in a layered way, i.e., like an onion, using the public keys of the servers in reverse order.
This encrypted `onion' is then sent to the first server, who batches a certain buffer and messages and then forwards this buffer in a random order, stripping one layer of encryption.
The following servers in the sequence repeat this process, until the final server sends the inside of the onion, the real message, to the recipient.
Implementations of mixnets that produces verifiably random permutations also exist.
See for example~\cite{bittauProchloStrongPrivacy2017} for an implementation of verifiable, oblivious shuffling using trusted hardware.

In conclusion, following also the discussion in~\cite{bittauProchloStrongPrivacy2017}, there are three main options for implementing a true, honest-but-curious, non-colluding shuffler.
The shuffler could be (1) a single trusted third-party; (2) a group of parties, in which trust is distributed; (3) one or more parties using trusted hardware.
The schemes as presented in this work are implementation-agnostic, i.e., they work with any choice of implementation.

\section{Security Proofs and Experiments}\label[appendix]{sec:security-proofs-experiments}

In this section, we provide security proofs for the three protocols, according to our definitions in \cref{sec:vldp-description}.
Before we detail the proofs, we provide the explicit experiments in each of the definitions and give some intuition in their construction.

\subsection{Experiments}\label[appendix]{subsec:experiments}
\cref{fig:completeness-experiment,fig:soundness-experiment,fig:zk-experiment,fig:shuffle-ind-experiment} describe the experiments used in the security definitions of \cref{subsec:security-defs}.
In all experiments, we explicitly describe the generation of the secret and public keys of each client's trusted environment on the second line of each experiment.
In reality, this is a step separate from our system, however, for completeness of the experiment definitions, we explicitly define it here.
Below, we give the formal definitions of these experiments, and provide some further intuition regarding their construction.

In the completeness experiment (\cref{fig:completeness-experiment}), we verify that the output $\tilde{x}_{i,j}$, with accompanying \nizkpk{} proof $\pi_{i,j}$, and public values $\tau_x^{i,j}$, generated by an honest client $i$ for time interval $j$, is accepted by an honest server, even when an adversary $\adv$ chooses the client's inputs $(x_{i,j},t_x^{i,j})$.\footnote{Completeness does not guarantee correctness of $\tilde{x}_{i,j}$. Correctness is implicitly guaranteed by soundness, which is defined below.}

\begin{figure}[H]
\centering
\scriptsize
\begin{pchstack}[boxed,center]
        \procedure[linenumbering,codesize=\scriptsize]{$\textbf{\textsf{Exp}}^\text{Comp}_{\adv}(1^\lambda, n, T)$}{
        \pp \leftarrow \setup(1^\secpar) \\
        \{\sk_i, \pk_i\}_i \leftarrow \sig.\kgen(\pp) \\
        (\ek, \vk, \pk_s, \sk_s, L) \leftarrow \kgen(\pp) \\
        \out_c^i \leftarrow \pcalgostyle{GenRand}(\pp, \textsf{aux}) \\
        \{x_{i,j}, t_x^{i,j}\}_{i,j} \leftarrow \adv(\pp, \ek, \vk, \pk_s, \sk_s, L, \{\pk_i\}_i) \\
        \pcif \exists i \in [n], j \in [T]: t_x^{i,j} \leq t_{j-1} \vee t_x^{i,j} > t_j \\
        \t \pcreturn \{\top\}_{i,j} \\
        \sigma_x^{i,j} \leftarrow \sig.\sign_{\sk_i}(x_{i,j}||t_x^{i,j}) \\
        \tilde{x}_{i,j}, \pi_{i,j}, \tau_x^{i,j} \leftarrow \pcalgostyle{Randomize}(\pp, \ek, t_j, \out_c^i, x_{i,j}, t_x^{i,j}, \sigma_x^{i,j}) \\
        \pcreturn \{\verify(\pp, \vk, t_j, \tilde{x}_{i,j}, \pi_{i,j}, \tau_x^{i,j})\}_{i,j}
    }
\end{pchstack}
\caption{Experiment for completeness definition.}
\label{fig:completeness-experiment}
\Description{Details of the experiment for completeness definition.}
\end{figure}

The soundness experiment (\cref{fig:soundness-experiment}) guarantees that no malicious, possibly colluding, clients are able to return a value $\tilde{x}_{i,j}$ that is not an honest evaluation of $\ldp.\apply(x_{i,j};\rho_{i,j})$, for a truly random, independently sampled $\rho_{i,j}$, for some $i$ and $j$.
In this experiment, the adversary is allowed to choose $t_x^{i,j}$ and controls all clients, who may deviate from the protocol arbitrarily.
The goal of the adversary is to let the server accept a tuple $(\tilde{x}_{i,j}, \pi_{i,j}, \tau_x^{i,j})$, where $\tilde{x}_{i,j}$ is not honestly computed.

\begin{figure}[H]
\centering
\scriptsize
\begin{pchstack}[boxed,center]
    \procedure[linenumbering,codesize=\scriptsize]{$\textbf{\textsf{Exp}}^\text{Snd-Real}_{\adv, S^*}(1^\lambda, n, T, \{x_{i,j}\}_{i,j})$}{
        \pp \leftarrow \setup(1^\secpar) \\
        \{\sk_i, \pk_i\}_i \leftarrow \sig.\kgen(\pp) \\
        (\ek, \vk, \pk_s, \sk_s, L, \{\pk_i\}_i) \leftarrow \kgen(\pp) \\
        \{t_x^{i,j}\}_{i,j} \leftarrow \adv(\pp, \ek, \vk, \pk_s, \{x_{i,j}\}_{i,j}) \\
        \pcif \exists i \in [n], j \in [T]: t_x^{i,j} \leq t_{j-1} \vee t_x^{i,j} > t_j \\
        \t \pcreturn \{\bot\}_{i,j} \\
        \sigma_x^{i,j} \leftarrow \sig.\sign_{\sk_i}(x_{i,j}||t_x^{i,j}) \\
        \{(\tilde{x}_{i,j}, \pi_{i,j}, \tau_x^{i,j})\}_{i,j} \leftarrow \adv^{S^*}(\pp, \ek, \vk, \pk_s, \{x_{i,j}, t_x^{i,j}, \sigma_x^{i,j}, \pk_i\}_{i,j})\\
        \pcreturn \{\verify(\pp, \vk, t_j, \tilde{x}_{i,j}, \pi_{i,j}, \tau_x^{i,j})\}_{i,j}
    }
\end{pchstack}
\caption{Experiment for soundness definition.}
\label{fig:soundness-experiment}
\Description{Details of the experiment for soundness definition.}
\end{figure}

The experiments for zero-knowledge (\cref{fig:zk-experiment}), define two different worlds. \textsf{Zk-real} denotes the real world, in which the adversary $\adv$ acts as the server, and interacts with honest clients (emulated by the environment).
In \textsf{Zk-sim}, $\adv$ acts as a server also, but instead interacts with a simulator $\sdv$.
$\sdv$ simulates an honest client, and should be able to generate messages with the same distribution as an actual client would, but without access to the input values $(x_{i,j},t_x^{i,j})$.
The adversary wins this game, if it can distinguish between both worlds.

\begin{figure}[H]
\centering
\scriptsize
\begin{pcvstack}[space=0.1cm,boxed,center]
    \procedure[linenumbering,codesize=\scriptsize]{$\textbf{\textsf{Exp}}^\text{Zk-Real}_{\adv}(1^\lambda, n, T, \{x_{i,j}\}_{i,j})$}{
        \pp \leftarrow \setup(1^\secpar) \\
        \{\sk_i, \pk_i\}_i \leftarrow \sig.\kgen(\pp) \\        (\ek, \vk, \pk_s, \sk_s, L, \trap) \leftarrow \adv(\pp) \\
        \out_c^i \leftarrow \pcalgostyle{GenRand}^{\adv}(\pp, \textsf{aux}) \\
        \{t_x^{i,j}\}_{i,j} \leftarrow \adv(\pp, \ek, \vk, \pk_s, \sk_s, L) \\
        \sigma_x^{i,j} \leftarrow \sig.\sign_{\sk_i}(x_{i,j}||t_x^{i,j}) \\
        \{(\tilde{x}_{i,j}, \pi_{i,j}, \tau_x^{i,j})\}_{i,j} \leftarrow \pcalgostyle{Randomize}(\pp, \ek, t_j, \out_c^i, x_{i,j}, \sigma_x^{i,j}) \\
        \pcif \text{\trap is not valid trapdoor for $(\relation, \ek, \vk)$} \\
        \t \t \t \vee \exists i \in [n], j \in [T]: t_x^{i,j} \leq t_{j-1} \vee t_x^{i,j} > t_j \\
        \t \pcreturn \bot \\
        \adv \leftarrow \{\tilde{x}_{i,j}, \pi_{i,j}, \tau_{i,j}\}_{i,j} \\
        \pcreturn (\textsf{view}_\adv, \{\tilde{x}_{i,j}\}_{i,j})
    }

    \procedure[linenumbering,codesize=\scriptsize]{$\textbf{\textsf{Exp}}^\text{Zk-Sim}_{\adv,\sdv}(1^\lambda, n, T, \{y_{i,j}\}_{i,j}))$}{
        \pp \leftarrow \setup(1^\secpar) \\
        \{\sk_i, \pk_i\}_i \leftarrow \sig.\kgen(\pp, \{\pk_i\}_i) \\
        (\ek, \vk, \pk_s, \sk_s, L, \trap) \leftarrow \adv(\pp, \{\pk_i\}_i) \\
        \out_c^i \leftarrow \sdv_1^{\adv}(\pp, \pk_s) \\
        \{t_x^{i,j}\}_{i,j} \leftarrow \adv(\pp, \ek, \vk, \pk_s, \sk_s, L) \\
        \{(\pi_{i,j}, \tau_x^{i,j})\}_{i,j} \leftarrow \sdv_2(\pp, \ek, \trap, t_j, \out_c^i, y_{i,j}) \\
        \pcif \text{\trap is not valid trapdoor for $(\relation, \ek, \vk)$} \\
        \t \t \t \vee \exists i \in [n], j \in [T]: t_x^{i,j} \leq t_{j-1} \vee t_x^{i,j} > t_j \\
        \t \pcreturn \bot \\
        \adv \leftarrow \{\tilde{y}_{i,j}, \pi_{i,j}, \tau_{i,j}\}_{i,j} \\
        \pcreturn (\textsf{view}_\adv, \{y_{i,j}\}_{i,j})
    }
    \end{pcvstack}
\caption{Experiments for the zero-knowledge definition.}
\label{fig:zk-experiment}
\Description{Details of the experiments for the zero-knowledge definition.}
\end{figure}

Finally, in the shuffle indistinguishability experiment (\cref{fig:shuffle-ind-experiment}), the adversary $\adv$ portrays the server and attempts to distinguish between two honest clients.
These honest clients, get the same input $(x,t_x)$, chosen by $\adv$, after both having executed the $\pcalgostyle{GenRand}$ protocol with $\adv$.
Subsequently, the environment only sends the outputs of $\pcalgostyle{Randomize}$ for one of the clients to $\adv$. $\adv$ wins the game if it is able to successfully determine to which client those outputs belong.

\begin{figure}[H]
    \centering
    \scriptsize
    \begin{pchstack}[boxed, center]
    \procedure[linenumbering,codesize=\scriptsize]{$\textbf{\textsf{Exp}}_\adv^\text{Sh-Ind}(\secpar)$}{
        \pp \leftarrow \setup(1^\secpar) \\
        \{\sk_i, \pk_i\}_i \leftarrow \sig.\kgen(\pp) \\
        (\ek, \vk, \pk_s, \sk_s, L, \trap) \leftarrow \adv(\pp, \{\pk_i\}_i) \\
        \adv(\pp) \rightarrow t_j\\
        \pcfor i \in \bin \\
        \t \pcalgostyle{GenRand}^\adv(\pp, t_j) \rightarrow \out_c^i \\
        \adv(\pp) \rightarrow (x, t_x) \\
        b \sample \bin \\
        \sig.\sign_{\sk_b}(x, t_x) \rightarrow \sigma_x^b \\
        \pcalgostyle{Randomize}(\pp, \ek, t_j, \out_c^b, x, t_x, \sigma_x^b) \rightarrow (\tilde{x}^b, \pi^b, \tau_x^b) \\
        \adv(\tilde{x}^b, \pi^b, \tau_x^b) \rightarrow b' \\
        \pcif \text{\trap is not valid trapdoor for $(\relation, \ek, \vk)$} \\
        \t \t \t \vee \exists i \in [n], j \in [T]: t_x^{i,j} \leq t_{j-1} \vee t_x^{i,j} > t_j \\
        \t \pcreturn \bot \\
        \pcreturn b' = b
    }
    \end{pchstack}
\caption{Experiment for shuffle indistinguishability.}
\label{fig:shuffle-ind-experiment}
\Description{Details of the experiment for shuffle indistinguishability.}
\end{figure}

\subsection{Proofs}\label[appendix]{subsec:proofs}
We only provide a full proof for the \pcalgostyle{Shuffle} scheme, as this is our main result.
Due to space constraints and similarity to the proof for the \pcalgostyle{Shuffle} scheme, we only provide sketches of the security proofs for the other schemes (\cref{thm:base,thm:expand}).

\begin{theorem}\label{thm:shuffle}
    $\mathcal{VLDP}_\text{shuffle}$ satisfies completeness, soundness, zero-knowledgeness and shuffle indistinguishability, given that $\nizkpk$ is secure for $\relation_\text{Base}$, $\comm$ a secure commitment scheme, $\prf$ a secure pseudo-random function, and $\sig$ an $\textsf{EUF-CMA}$ secure digital signature scheme.
\end{theorem}
\begin{proof} We prove the properties one by one:
\subsubsection*{(1) Completeness} We see that the scheme satisfies completeness as long as the $\pcalgostyle{Randomize}_\text{shuffle}$ procedure outputs a valid \nizkpk{} proof $\pi$ for $\relation_\text{shuffle}$.
It therefore suffices to show that the proof for $\relation_\text{shuffle}$ is correctly computed when all parties are honest.

To prove this, fix an arbitrary $i \in [n]$ and an arbitrary $j \in [T]$.
First we observe that statement 1 of $\relation_\text{shuffle}$ has to hold by the condition on $t_x^{i,j}$.
Statement 2 of $\relation_\text{shuffle}$ holds by construction of $\sigma_x^{i,j}$.
Also, statements 3 and 5 hold, because these correspond exactly to the computations done in $\pcalgostyle{GenRand}_\text{shuffle}$ by the honest parties.
Finally, statements 4, 6, and 7 hold by the fact that an honest party will evaluate these correctly in $\pcalgostyle{Randomize}_\text{shuffle}$.

Therefore, we can conclude that the proof $\pi$ will be verified successfully, because our \nizkpk scheme is complete itself, i.e., $\textsf{Verify}_\text{shuffle} \neq \bot$, except for some probability that is $\negl$.
Since, $i$ and $j$ were picked arbitrarily, and both $T$ and $n$ are $\poly$, we can conclude that the total probability is also $\negl$.
        
\subsubsection*{(2) Soundness}
        To prove soundness, we define a series $\textbf{\textsf{Exp}}_0$--$\textbf{\textsf{Exp}}_3$ of hybrid experiments, where $\textbf{\textsf{Exp}}_0$ is $\textbf{\textsf{Exp}}^\text{Snd-Real}_{\adv, S^*}(1^\lambda, n, T, \{x_{i,j}\}_{i,j})$ (from \cref{def:soundness} as defined in \cref{fig:soundness-experiment}) and $\textbf{\textsf{Exp}}_3$ is close to the ideal.
        Recall, that by definition of soundness, the adversary $\adv$, who controls all, potentially malicious, clients, can send messages to and receive corresponding replies from an honest server $\sdv^*$.
        We will show that all these experiments are (computationally) indistinguishable, and therefore $\mathcal{VLDP}_\text{shuffle}$ is sound.

        \begin{itemize}[leftmargin=*]
        \item $\textbf{\textsf{Exp}}_0$: Equal to $\textbf{\textsf{Exp}}^\text{Snd-Real}_{\adv, S^*}(1^\lambda, n, T, \{x_{i,j}\}_{i,j})$.
        
        \item $\textbf{\textsf{Exp}}_1$: This is the same experiment as $\textbf{\textsf{Exp}}_0$, except that now we run a p.p.t.\ knowledge extractor $\edv_\adv$ to obtain the witness $\vec{w}$ that $\adv$ used to generate the proof.
        We know that such an extractor exists, due to knowledge soundness of \nizkpk.
        Now, instead of checking a proof $\pi_{i,j}$, the verifier uses $\vec{w}_{i,j}$ and $\vec{\phi}_{i,j}$ to check the statements of $\relation_\text{shuffle}$ directly.
        Clearly, both games are identical up to the probability that $\adv$ wins the knowledge soundness game for one of the proofs.
        By knowledge soundness of $\nizkpk$ we know that this probability is negligible: \[|\Pr[\textbf{\textsf{Exp}}_0] - \Pr[\textbf{\textsf{Exp}}_1] | \leq \negl.\]
        
        \item $\textbf{\textsf{Exp}}_2$: This is the same experiment as $\textbf{\textsf{Exp}}_1$, except that now the verifier asserts that the values $x$ and $t_x$ in $\vec{w}_{i,j}$ are equal to $x_{i,j}$ and $t_x^{i,j}$.
        If this is not the case, we set $\textsf{fail}_2 = \textsf{true}$.
        Clearly both games are identical up to $\textsf{Fail}_2$:
        \[|\Pr[\textbf{\textsf{Exp}}_1] - \Pr[\textbf{\textsf{Exp}}_2] | \leq \Pr[\textbf{\textsf{Fail}}_2].\]
        Since $\sig$ is an $\textsf{EUF-CMA}$ secure signature scheme that has been used to generate $\sigma_x$ and $\adv$ does not know $\sk_i$, it follows that $\Pr[\textbf{\textsf{Fail}}_2] \leq \negl$.
        
        \item $\textbf{\textsf{Exp}}_3$: This is the same experiment as $\textbf{\textsf{Exp}}_2$, except that the server $S^*$ now additionally maintains a list $R$ containing entries of the form $(i, (\pk_i, \cm_{k_c}^i, k_s^i))$.
        These entries correspond to messages $(\pk_i, \cm_{k_c})$ received during occurrences of the $\pcalgostyle{GenRand}_\text{shuffle}$ protocol with user $i$. $k_s^i$ corresponds to the server seed $k_s$ that was generated by the server during this particular occurrence.
        Note, that each user $i$ can only have one entry in $R$, due to step 4 in $\pcalgostyle{GenRand}_\text{shuffle}$.
        Now, if $\textsf{fail}_2$ has not been set, the server asserts, in $\pcalgostyle{Verify}_\text{shuffle}$, that $R$ indeed has an entry $(\star, (\pk_i, \cm_{k_c}^i, k_s^i))$, where $(\pk_i, \cm_{k_c}^i, k_s^i)$ are the corresponding elements of $\vec{w}_{i,j}$.
        If this is not the case, we set $\textsf{fail}_3 = \textsf{true}$.
        Clearly both games are identical up to $\textsf{Fail}_3$:
        \[|\Pr[\textbf{\textsf{Exp}}_2] - \Pr[\textbf{\textsf{Exp}}_3 | \leq \Pr[\textbf{\textsf{Fail}}_3].\]
        The only way, by which the event $\textbf{\textsf{Fail}}_3$ can occur, is if the client can produce a tuple $(\pk_i, \cm_{k_c}^i, k_s^i)$ with signature $\sigma_s^i$, such that $\sig.\verify_{\pk_s}(\sigma_s^i, \pk_i || \cm_{k_c}^i || k_s^i) = 1$.
        $\sdv^*$ will only have generated one signature for each $\pk_i$, due to step 4 in $\pcalgostyle{GenRand}_\text{shuffle}$, and this tuple is on $R$.
        Therefore, for $\textbf{\textsf{Fail}}_3$ to occur, the client must have forged a signature on a tuple $(\pk_i', \cm_{k_c}^{i\prime}, k_s^{i\prime})$, where at least one element differs from $(\pk_i, \cm_{k_c}^i, k_s^i)$.
        Since $\sig$ is an $\textsf{EUF-CMA}$ secure signature scheme that has been used to generate $\sigma_s^i$ and $\adv$ does not know $\sk_s$, it follows that $\Pr[\textbf{\textsf{Fail}}_3] \leq \negl$.
        \end{itemize}

        Finally, we observe that if $\textsf{Fail}_3$ does not occur, we are essentially at a point where $\tilde{x}_{i,j} = \ldp.\apply(x_{i,j}; \rho_{i,j})$, where $\rho_{i,j} = \prf(k_c^i \oplus k_s^i, s_j)$.
        We observe that $k_s$ is chosen uniformly at random and independently of $x_{i,j}$.
        Moreover, $k_s^i$ is bound to all $x_{i,j}$, since the same public key $\pk_i$ is used inside $\sigma_s$ and for the verification of $\sigma_x$.
        Also, $s_j$ is public and independent of all $x_{i,j}$ and all $x_{i,j}$ are given.
        Next to this, $k_c$ is uniquely determined by $\cm_{k_c}$, except with negligible probability, according to the binding property of $\comm$.
        And, since $\cm_{k_c}^i$ is fixed before $k_s^i$ is chosen uniformly at random, $k_c^i \oplus k_s^i$ is also be a uniform random bitstring, and by the definition of a secure $\prf$, $\rho_{i,j}$ will also be distributed at random, except with negligible probability.
        Therefore, it follows that:
        \begin{multline*}
            \left\vert \Pr[\textbf{\textsf{Exp}}_3] - \Pr\left[\ldp.\apply(x_{i,j}; \rho_{i,j}) = \{y_{i,j}\}_{i,j} \middle\vert \right.\right. \\
            \left.\left. \rho_{i,j} \sample \{0,1\}^* \right]\right\vert \leq \negl.
        \end{multline*}

\subsubsection*{(3) Zero-Knowledge}
To show that $\mathcal{VLDP}_\text{shuffle}$ satisfies the zero-knowledge property, we will show that the joint distribution of the output and all messages received by $\adv$ in the real scheme is indistinguishable from those generated by the simulator $\sdv = (\sdv_1, \sdv_2)$ (see \cref{fig:simulators-shuffle}).
Note, that in our model we assume that the verifier behaves like an honest-but-curious adversary, and thus follows the protocol.
The first simulator algorithm $\sdv_1$, simulates a client in $\pcalgostyle{GenRand}_\text{shuffle}$, thereby also interacting with $\adv$, representing the server.
The second $\sdv_2$, simulates the $\pcalgostyle{Randomize}_\text{shuffle}$ algorithm.

\begin{figure}[H]
\begin{pchstack}[boxed,center,space=0.25cm]
    \scriptsize
    \procedure[linenumbering,codesize=\scriptsize]{$\sdv_1(\pp, \pk_s)$}{
        k_c \sample \bin^* \\
        r_{k_c} \sample \bin^* \\
        \cm_{k_c} = \comm(k_c; r_{k_c} )\\
        \text{Send } (\pk_i, \cm_{k_c}^i) \text{to \adv{} as client $i$.} \\
        \text{Receive $(k_s, \sigma_s)$ from $\adv$.} \\
        \pcreturn (k_c, r_{k_c}, \cm_{k_c}, k_s, \sigma_s)
    }
    \procedure[space=auto,linenumbering]{$\sdv_2(\pp, \ek, \trap, t_j, \out_c^i, y_{i,j})$}{
        \vec{\phi} = (t_{j-1}, t_j, \pk_s, s_j, y_{i,j}) \\
        \pi \leftarrow \simulator_\text{nizk}(\relation, \trap, \vec{\phi})\\
        \pcreturn (\pi, y_{i,j})
    }
\end{pchstack}
\Description{Details of both simulator algorithms.}
\caption{Simulator $\sdv = (\sdv_1, \sdv_2)$ for the zero-knowledge proof of \cref{thm:shuffle}.}
\label{fig:simulators-shuffle}
\end{figure}

First, we observe that the message produced by $\sdv_1$ is distributed as in the real experiment, since our server is honest-but-curious and $\sdv_1$, follows the exact same steps as $\pcalgostyle{GenRand}_\text{shuffle}$.
Second, we observe that $y_{i,j}$ is fixed.
Third, we observe that the values in $\vec{\phi}$ are either public or fixed, and therefore are indistinguishable between the real world and the simulator.
Since $\nizkpk$ is a secure scheme for $\relation_\text{shuffle}$ with the zero-knowledge property, and $\trap$ is a valid trapdoor, $\sdv_2$ can use the $\nizkpk$ simulator $\simulator$ to generate a proof $\pi$ with the correct distribution, i.e., such that $\pi$ passes verification for $\vec{\phi}$.
From these observations it follows that the joint distribution of the real scheme and its output, is indistinguishable from that generated by \sdv, and therefore the scheme is zero-knowledge.

\subsubsection*{(4) Shuffle Indistinguishability}
        To prove shuffle indistinguishability, we describe a series of hybrid experiments $\textbf{\textsf{Exp}}_0$ -- $\textbf{\textsf{Exp}}_3$, where $\textbf{\textsf{Exp}}_0$ is equal to $\textbf{\textsf{Exp}}^\text{Sh-Ind}_\adv$, and $\textbf{\textsf{Exp}}_3$ leaves $\mathcal{A}$ essentially random guessing.
        We will show that all these experiments are indistinguishable, and therefore $\mathcal{VLDP}_\text{shuffle}$ satisfies shuffle indistinguishability.

        \begin{itemize}
            \item $\textbf{\textsf{Exp}}_0$: Equal to $\textbf{\textsf{Exp}}^\text{Sh-Ind}_\adv$ in \cref{def:shuffle-indistinguishability}.
            
            \item $\textbf{\textsf{Exp}}_1$: This is the same experiment as $\textbf{\textsf{Exp}}_0$, except that now we also give $\trap$ as input to $\pcalgostyle{Randomize}_\text{shuffle}$ and replace the computation of $\pi^b$ (step 6 in $\pcalgostyle{Randomize}_\text{shuffle}$), by a simulated proof using the $\nizkpk$ simulator $\simulator$ using $\trap$.
            By the zero-knowledge property of $\nizkpk$ we conclude that both experiments are indistinguishable: \[|\Pr[\textbf{\textsf{Exp}}_0] - \Pr[\textbf{\textsf{Exp}}_1] | \leq \negl.\]
            
            \item $\textbf{\textsf{Exp}}_2$: This is the same experiment as $\textbf{\textsf{Exp}}_1$, except that in step 1 of $\pcalgostyle{Randomize}_\text{shuffle}$ we replace $k_c$ by a random bit-string of the same length.
            We will still pass verification, since $\pi$ has been replaced by a simulated proof.
            Finally, by the hiding property of $\comm$, \adv{} cannot distinguish between the usage of $k_c$ or some other random bitstring of the length, except with negligible probability.
            Therefore, we conclude that both experiments are indistinguishable: \[|\Pr[\textbf{\textsf{Exp}}_1] - \Pr[\textbf{\textsf{Exp}}_2] | \leq \negl.\]
            
            \item $\textbf{\textsf{Exp}}_3$: This is the same experiment as $\textbf{\textsf{Exp}}_2$, except that we replace $\rho$ in step 2 of $\pcalgostyle{Randomize}_\text{shuffle}$ by a random bit-string of the same length.
            We will still pass verification, since $\pi$ has been replaced by a simulated proof.
            Finally, since $k_c$ was already replaced by a uniform random bitstring, we know that $k$ is a uniform random bitstring, and by the fact that $\prf$ is secure, $\prf(k, s_j)$ is indistinguishable from a random bit-string of the same length.
            Therefore, we conclude that both experiments are indistinguishable: \[|\Pr[\textbf{\textsf{Exp}}_2] - \Pr[\textbf{\textsf{Exp}}_3] | \leq \negl.\]
        \end{itemize}
        We observe that $\pi^b$ is replaced by a simulated proof in $\textbf{\textsf{Exp}}_3$, i.e., $\pi^0$ has the same distribution as $\pi^1$.
        Moreover, $\tilde{x}^b = \ldp.\apply(x; r)$, where $x$ is the same for both values of $b$ and $r$ is chosen uniformly at random, i.e., $\tilde{x}^0$ has the same distribution as $\tilde{x}^1$.
        Finally, $\tau_x^b$ is empty.
        Therefore, $(\tilde{x}^b, \pi^b, \tau_x^b)$ has the same distribution for either value of $b$, meaning the advantage of $\adv$ in $\textbf{\textsf{Exp}}_3$ is essentially the same as if $\adv$ were random guessing.
        Thus, we conclude that \[|\Pr[\textbf{\textsf{Exp}}_3] - \frac{1}{2}| \leq \negl. \qedhere \]
\end{proof}

\begin{theorem}\label{thm:base}
    $\mathcal{VLDP}_\text{base}$ satisfies completeness, soundness, and zero-knowledgeness, given that $\nizkpk$ is secure for $\relation_\text{Base}$, $\comm$ a secure commitment scheme, $\prf$ a secure pseudo-random function, and $\sig$ an $\textsf{EUF-CMA}$ secure digital signature scheme.
\end{theorem}
\begin{proof}[Proof (Sketch)] We prove the properties one by one:
\subsubsection*{(1) Completeness} By inspection of the protocol and completeness of \nizkpk.
\subsubsection*{(2) Soundness} Also here, we define a series of hybrid experiments, where $\textbf{\textsf{Exp}}_0$ is equal to $\textbf{\textsf{Exp}}^\text{Snd-Real}_{\adv, S^*}$ and $\textbf{\textsf{Exp}}_3$ is close to the ideal.
        In $\textbf{\textsf{Exp}}_1$, we again use the knowledge extractor $\edv_\adv$ to obtain the witness $\vec{w}$ and verify the statements in $\relation_\text{base}$ directly.
        $\textbf{\textsf{Exp}}_2$ is analogous to that in the proof for \cref{thm:shuffle}.

        $\textbf{\textsf{Exp}}_3$ is similar to that in the proof for \cref{thm:shuffle}, however, $R$ will now contain entries $((i, t_j), (\pk_i, \cm_{\rho_c}^{i,j}, k_s^{i,j}))$.
        I.e., each user $i$ now has one entry for each $t_j$, rather than using the same entry for all $t_j$.
        Analogously to before, the server now verifies, in $\pcalgostyle{Verify}_\text{base}$ that $((\star, t_j), (\pk_i, \cm_{k_c}^{i,j}, k_s^{i,j})$ is on $R$.

        Finally, analogous to the proof for \cref{thm:shuffle}, by the binding property of $\comm$, we know that $\rho$ is sampled independently and uniformly at random and, irrespective of \adv.
        Therefore, we can conclude soundness.

\subsubsection*{(3) Zero-Knowledge}
        Just like in the proof for \cref{thm:shuffle}, $\sdv_1$ is identical to $\pcalgostyle{GenRand}_\text{base}$, i.e., for all $t_j$ it computes $\cm_{\rho_c}$, and receives $(k_s, \sigma_s)$ from $\adv$.

        Just like in $\sdv_2$, we also create a simulated proof from the statement vector $\vec{\phi}$ alone.
        $\vec{\phi}$ consist of $\sdv_2$'s inputs, values from $\sdv_1$ and $\rho_s$.
        $\sdv_2$ simply computes $\rho_s$ as in the real case.
        Note, that also here we use $y_{i,j}$ for $\tilde{x}$, rather than computing it from some signed input $x$ and the randomness $\rho$.

        Additionally, unlike the proof for \cref{thm:shuffle}, $\sdv_2$ outputs the values $(\pk_i, \cm_{\rho_c}, k_s, \sigma_s)$, where $\pk_i$ is known and fixed, and the other values come from $\sdv_1$.

        Following arguments analogous to the proof for \cref{thm:shuffle} and due to zero-knowledgeness of $\nizkpk$ we conclude that $\textsf{Base}$ is zero-knowledge.
\end{proof}

\begin{theorem}\label{thm:expand}
    $\mathcal{VLDP}_\text{expand}$ satisfies completeness, soundness, and zero-knowledgeness, given that $\nizkpk$ is secure for $\relation_\text{Base}$, $\comm$ a secure commitment scheme, $\prf$ a secure pseudo-random function, $\sig$ an $\textsf{EUF-CMA}$ secure digital signature scheme, and $\crh$ (use to construct $\merkletree$) a collision-resistant hash function.
\end{theorem}
\begin{proof}[Proof (Sketch)] We prove the properties one by one:

\subsubsection*{(1) Completeness} By inspection of the protocol and completeness of \nizkpk.

\subsubsection*{(2) Soundness} Also here, we define a series of hybrid experiments, where $\textbf{\textsf{Exp}}_0$ is equal to $\textbf{\textsf{Exp}}^\text{Snd-Real}_{\adv, S^*}$ and $\textbf{\textsf{Exp}}_3$ is close to the ideal.
        In $\textbf{\textsf{Exp}}_1$, we again use the knowledge extractor $\edv_\adv$ to obtain the witness $\vec{w}$ and verify the statements in $\relation_\text{expand}$ directly.
        $\textbf{\textsf{Exp}}_2$ is analogous to the proof for \cref{thm:shuffle}.

        $\textbf{\textsf{Exp}}_3$ is analogous to that in the proof for \cref{thm:shuffle}, however, $\cm_{\rho_c}$ is replaced by $\rt$.

        Finally, analogous to the proof for \cref{thm:shuffle}, by the binding property of $\comm$, and by collision resistance of the hash function $\crh$ used to construct the $\merkletree$, we know that $\rho$ is uniformly at random, irrespective of $\adv$.
        Thus, we conclude that $\textsf{Expand}$ is sound.

\subsubsection*{(3) Zero-Knowledge}
        Just like in the proof for \cref{thm:shuffle}, $\sdv_1$ acts identical to $\pcalgostyle{GenRand}_\text{base}$, i.e., for all $t_j$ it computes $\rt$, and receives $(k_s, \sigma_s)$ from $\adv$.

        Just like in $\sdv_2$, we also create a simulated proof from the statement vector $\vec{\phi}$ alone.
        $\vec{\phi}$ consist of $\sdv_2$'s inputs, values from $\sdv_1$ and $\rho_s$.
        $\sdv_2$ simply computes $\rho_s$ as in the real case.
        Note, that also here we use $y_{i,j}$ for $\tilde{x}$, rather than computing it from some signed input $x$ and the randomness $\rho$.
        Additionally, unlike the proof for \cref{thm:shuffle}, $\sdv_2$ outputs the values $(\pk_i, \rt, k_s, \sigma_s)$, where $\pk_i$ is known and fixed, and the other values come from $\sdv_1$.

        Following arguments analogous to the proof for \cref{thm:shuffle} and due to zero-knowledgeness of \nizkpk we obtain the zero-knowledge property for $\textsf{Base}$.
\end{proof}

\section{Appendix to Section 7}
\label[appendix]{sec:appendix-to-section-7}

\subsection{Dataset Description}
\label[appendix]{subsec:dataset-description}
\subsubsection*{Dataset 1: Geolife GPS Trajectory} This is a location dataset of 182 users, collected as part of Microsoft Research Asia's GeoLife project over the period of 2007 to 2012.\footnote{The original dataset can be found at \url{https://www.microsoft.com/en-us/research/publication/geolife-gps-trajectory-dataset-user-guide/}} Each data point contains latitude, longitude (GPS coordinates) and altitude information of a user on a given day.
Upon inspecting the data, we found that it was highly sparse.
In particular, only a small subset of users had GPS coordinates recorded for any given day.
We therefore decided to extract five readings from each user from five different days, assuming that the corresponding readings were taken from the same day.
For each day, we only took the first GPS coordinates.
We then used the Nominatim API\footnote{See \url{https://nominatim.org/}} to obtain the address corresponding to each GPS coordinate using reverse geocoding.
From the address thus returned, we retained only the postcode of each location.
Most of the postcodes were only visited by a very small amount of users.
We therefore took the 7 top postcodes and included the rest into a single postcode named \verb+all_others+.
Thus, in total we have 8 postcodes per day.
This dataset is used for the LDP algorithm for histograms where the goal is to estimate the true histogram of users in each postcode per day.
Note that we have $k = 8$, where $\{1,\ldots,k\}$ represent the respective postcodes in the algorithm for histograms (see \cref{fig:ldp-algos}).

\subsubsection*{Dataset 2: Smart Meter}
This dataset is extracted from the smart meters in London dataset.\footnote{See
\url{https://www.kaggle.com/datasets/jeanmidev/smart-meters-in-london}} Which was originally extracted from energy readings dataset of 5,567 London households, which were obtained during the Low Carbon London project led by UK Power Networks between 2011 and 2014.
More specifically, we took the \verb+daily_dataset.csv+ file and take the mean energy reading from each household for the last five days of this dataset with the last date 25/02/2014.
Households that did not have a mean energy reading for a particular day are assigned mean energy of 0 for that day.
Finally, we normalized the readings by dividing each mean energy value by the maximum mean energy in \verb+daily_dataset.csv+ which was 6.928 kWh. This dataset is used for the LDP algorithms for reals where the goal is to estimate the average energy reading of a household per day, by estimating the sum of mean energy consumption across all households and then dividing it by the number of households, i.e., 5,567.
We use a precision level of $k = 10$ (see the algorithm for reals in \cref{fig:ldp-algos}) for our experiments.

\subsection{Possible Optimizations and Alternatives}
\label[appendix]{subsec:optimizations-alternatives}
While the experiments in \cref{sec:evaluation} show the practicality of our schemes, certain optimizations and/or different choices could be made with respect to our implementation.
Specifically, we discuss how different choices would influence the trade-off between security, efficiency, and practicality.

Our current choices were based on primitives that provide a strong level of security (targeting 128 bits) within practical times.
Given that our results show that performance is very practical, efficiency improvements are not a requirement for practical adoption, however may still be considered depending on the specific requirements of the use case.

\subsubsection*{SNARK-friendly CRH}
In our current implementation, we rely on the Blake2s CRH inside our \prf, \sig scheme, and \pcalgostyle{MerkleTree}.
Blake2s is a standardized scheme and its security level has been well vetted and confirmed.
Moreover, it is more efficient to encode inside a zk-SNARK circuit than alternatives with a similar security level, e.g., SHA256.
However, in some recent works we observe an increased interest in so-called SNARK-friendly hash functions, often algebraic hash functions that can be more efficiently computed inside a circuit.
Examples of such schemes are Poseidon~\cite{grassiPoseidonNewHash2021}, MiMC~\cite{albrechtMiMCEfficientEncryption2016}, and Pedersen~\cite{hopwoodZcashProtocolSpecification2023} hashes.
These schemes can reduce prover times by as much as a factor of 10~\cite{steidtmannBenchmarkingZKCircuitsCircom2023}.
However, this often comes at the cost of reduced security.
Poseidon and MiMC are novel constructs and have not yet been properly vetted by the community.
This novelty, in combination with their algebraic construction, might give rise to unforeseen attacks, such as the algebraic attacks shown against MiMC~\cite{eichlsederAlgebraicAttackCiphers2020}.
The Pedersen hash on the other hand is known to be secure, but has the downside that it relies on the discrete log assumption, whereas Blake2s requires no such assumption.
We can safely use the Pedersen hash inside our Merkle tree, since breaking the discrete log assumption would require efforts similar to breaking the zk-SNARK scheme we use.

However, we have chosen to not use Pedersen hashes for our \prf and \sig schemes.
The \prf is evaluated out-of-circuit in the \textsf{Base} and \textsf{Expand} schemes anyhow, and we chose to use the same primitive in the \textsf{Shuffle} scheme for comparability.

For the \sig schemes, we chose to use commonly available primitives, as in many use cases the trusted environment will not provide more novel or custom primitives, such as Pedersen or Poseidon hash functions.
However, in cases where such primitives are available, they could be used to improve efficiency at the cost of (slightly) decreased security and general applicability.

\subsubsection*{SNARK-friendly signatures}
Next to using a different hash function to transform the input message to a fixed digest, we could have also used a different signature scheme than Schnorr's.
However, it should be noted that Schnorr is commonly available and very efficient.
E.g., it is around 2-2.5x times faster than EdDSA inside a zk-SNARK circuit~\cite{steidtmannBenchmarkingZKCircuitsCircom2023}.
Although, due to the small number of constraints, the impact on the total performance of our scheme will be small.

\subsubsection*{zk-SNARK}
The predominant component for our computation times is determined by the zk-SNARK scheme and the way our constraints are encoded inside the zk-SNARK circuit.

First, we note that we implemented our constraint circuit using readily available `gadgets' from the Arkworks library.
While these gadgets are of good quality and are implemented efficiently, we do get some more constraints than are strictly required in hand-optimized constraint systems.
These constraints might possibly be reduced by manual inspection, however due to the large number of constraints we expect this to be a tedious task, for which we only get an (almost) negligible increase in performance.
Moreover, manual optimization of these constraints would reduce modularity of our software implementation, which may be a significant practical drawback.

Next to this, one could also consider using an alternative scheme to Groth16.
For example, schemes with a \emph{universal setup} (e.g., Marlin~\cite{chiesaMarlinPreprocessingZkSNARKs2020}), \emph{post-quantum security} (e.g., Fractal~\cite{chiesaFractalPostquantumTransparent2020}) or a \emph{transparent setup} (e.g., Fractal~\cite{chiesaFractalPostquantumTransparent2020} or SuperSonic~\cite{bunzTransparentSNARKsDARK2020}) could be employed.
Generally speaking, these alternatives are less efficient than Groth16, with respect to proof generation and verification time.
However, in some use cases, universal setups can be used to prevent the server from having to run one setup per LDP algorithm.
In practice, however we do not expect this trade-off to be beneficial.
Moreover, the removal of a trusted setup is not necessary in our case, since we assume the server to behave semi-honestly and not collude with any of the clients, i.e., it will not give the trapdoor to any of the clients.

\end{document}